\documentclass[11pt,a4paper,reqno]{article}
\usepackage{amssymb}
\usepackage{latexsym}
\usepackage{amsmath}
\usepackage{graphicx}
\usepackage{relsize}
\usepackage{amsthm}
\usepackage{empheq}
\usepackage{bm}
\usepackage{booktabs}
\usepackage[dvipsnames]{xcolor}
\usepackage{pagecolor}
\usepackage{subcaption}
\usepackage{tikz,pgfplots,lipsum,lmodern}
\usepackage[most]{tcolorbox}

\newcommand{\dstirling}[2]{\genfrac{[}{]}{0pt}{0}{#1}{#2}}
\usepackage[margin=2.5cm]{geometry}

 \usepackage[section]{placeins}

\usepackage{hyperref}
\usepackage[pagewise]{lineno}
\numberwithin{equation}{section}
\theoremstyle{definition}
\newtheorem{theorem}{Theorem}[section]
\newtheorem{corollary}[theorem]{Corollary}
\newtheorem{proposition}[theorem]{Proposition}
\newtheorem{definition}[theorem]{Definition}

\newtheorem{lemma}[theorem]{Lemma}

\theoremstyle{remark}
\newtheorem{remark}{Remark}
\newtheorem{algorithm}{Algorithm}

\usepackage{calrsfs}

\usepackage{calrsfs}
\DeclareMathAlphabet{\pazocal}{OMS}{zplm}{m}{n}

\newcommand\qbin[3]{\left[\begin{matrix} #1 \\ #2 \end{matrix}\right]_{#3}}

\newcommand\bbql[1]{\textbf{v}_{m}^{\textnormal{L}}(#1)}
\newcommand\bbqlp[1]{\textbf{v}_{p}^{\textnormal{L}}(#1)}
\newcommand\ballL[1]{\mathbf{B}_{m}^{\textnormal{L}}(#1)}

\newcommand{\numberset}{\mathbb}

\newcommand{\Z}{\numberset{Z}}

\newcommand{\F}{\numberset{F}}

\newcommand{\mV}{\mathcal{V}}

\newcommand{\mC}{\mathcal{C}}

\newcommand{\mF}{\mathcal{F}}
\newcommand{\mW}{\mathcal{W}}

\newcommand{\mB}{\mathcal{B}}
\newcommand{\mE}{\mathcal{E}}

\renewcommand{\longrightarrow}{\to}

\newcommand{\wL}{w^{\textnormal{L}}}

\newcommand{\Zm}{\Z_{m}}
\newcommand{\Zps}{\Z_{p^{s}}}
\newcommand{\Zp}{\Z_{p}}

\newcommand*{\myproofname}{Proof of Lemma 3.7}

\makeatletter
\renewcommand*\env@matrix[1][*\c@MaxMatrixCols c]{%
  \hskip -\arraycolsep
  \let\@ifnextchar\new@ifnextchar
  \array{#1}}
\makeatother

\title{
\textbf{On the Number of $t$-Lee-Error-Correcting Codes}}

\usepackage{authblk}
\author[1]{Nadja Willenborg}
\author[1]{Anna-Lena Horlemann}
\affil[1]{University of St.Gallen, Switzerland}
\author[2]{Violetta Weger}
\affil[2]{Technical University of Munich, Germany}

\date{}

\begin{document}

\maketitle

\thispagestyle{empty}
	
\begin{abstract}
 We consider $t$-Lee-error-correcting codes of length $n$ over the residue ring $\Zm := \mathbb{Z}/m\mathbb{Z}$ and determine upper and lower bounds on the number of $t$-Lee-error-correcting codes. We use two different methods, namely estimating isolated nodes on bipartite graphs and the graph container method. The former gives density results for codes of fixed size and the latter for any size. This confirms some recent density results for linear Lee metric codes and provides new density results for nonlinear codes. To apply a variant of the graph container algorithm we also investigate some geometrical properties of the balls in the Lee metric.
\end{abstract}

\section{Introduction}
A classical discipline in coding theory is the study of the behaviour of random codes and computing the density of optimal codes with respect to some bounds. 
For Hamming metric codes over $\mathbb{F}_{q}$ we have several famous results, such as \cite{GVhigh, GVhigh2}, where the authors show that if we let the length $n$ grow, random codes attain with high probability the Gilbert-Varshamov bound \cite{gilbert1952comparison, varshamov1957estimate}. On the other hand, it is well known, that if we let $q$ grow, random codes are with high probability Maximum Distance Separable (MDS) codes, i.e., they attain the Singleton bound. Similar results have been obtained in the rank metric, namely for $n$ growing random rank-metric codes attain the rank-metric Gilbert-Varshamov bound \cite{loidreau} and we know that the density of Maximum Rank Distance (MRD) codes
depends on whether they are $\mathbb{F}_{q^m}$-linear or $\mathbb{F}_q$-linear.
For the former it was shown in \cite{neri2018genericity} that MRD codes are dense, while for the latter case 
it was shown that MRD codes are neither dense nor sparse \cite{antrobus2019maximal,byrne, gluesing, gruica2022common}. In this paper, we investigate such questions in a new setting: modules over integer residue rings endowed with the Lee metric.
The Lee metric has been introduced to cope with phase modulation, error correction in constrained memories, or in cryptographic applications (see for example \cite{fuleeca}).
Density questions for such Lee metric codes in $\mathbb{Z}_{p^s}^n$ have already been studied in \cite{byrne2022density,byrne2023bounds}. In particular, it has been shown that for $n$ growing, random codes attain the respective Gilbert-Varshamov bound with high probability and that optimal linear codes with respect to Lee metric Singleton bounds \cite{shiromoto, alderson} are sparse, for $n,s$ or $p$ growing. Sparseness has also been shown for optimal linear codes with respect to the Lee metric Plotkin bound. In this paper we confirm these results and use techniques from graph theory to obtain new results for nonlinear codes.

We consider two different techniques: bipartite graphs and the container method. 
The former technique is based on bounding the number of some non-isolated nodes in certain bipartite graphs whose right node set consists of codes in $\Zm^{n}$ of cardinality $S$ and whose left node set contains elements from $\Zm^{n}$ whose minimum weight is upper bounded. From this approach we first derive upper and lower bounds on the density of Lee metric codes. These bounds then give us asymptotic results in the nonlinear case, which was previously unknown.

The second technique uses a graph $G=(V,E)$ with node set $V=\Zm^{n}$ and edge set $E=\{(x,y) \in \Zm^{n} \times \Zm^{n}: x \neq y,  d^{\textnormal{L}}(x,y)
 \leq 2t \}$, where $ d^{\textnormal{L}}(x,y)$ is the Lee distance between $x$ and $y$. Note that $G$ is also the $2t$ graph power of the Cartesian product $C_{m} \Box C_{m} \Box \cdots \Box C_{m}$, where $C_{m}$ is the cycle of length $m$. The node set of this Cartesian product is $\Zm^n$ and there is an edge between $x,y \in \Zm^n$ if and only if their Lee distance is $1$ from each other, see also \cite{abiad2022independence, kim20112}. We show that when fixing $m$, the number of $t$-Lee-error-correcting codes of length $n$ is at most $2^{(1+o(1))H_{m}(n,t)}$, for all $t \leq 10\sqrt{\theta_{m} n}$. This problem has already been considered in \cite{dong2022number} and \cite{balogh2016applications} for Hamming metric codes, and also in \cite{balogh2016applications} for codes endowed with the Enomoto-Katona distance (also called transportation distance). The key tool for estimating $t$-error-correcting codes are so-called `container results' which state that the independent sets of a graph lie only in a small number of subsets of the node set which are `almost independent sets'. The estimate of $t$-error-correcting codes with distance $2t+1$ is equivalent to the estimate of independent sets in our graph. Although our graph does not have the properties desired by Sapozhenko, i.e., the graph is almost regular and an expander, we can transfer the saturation estimates from \cite{dong2022number} to the Lee metric and then apply the container result from \cite{balogh2016applications}. To do this, we first derive estimates of the volumes of balls in the Lee metric with different radii. We also characterize the intersection of two balls with equal Lee radii with distance $\ell$ between their respective centers. In particular we show that the size of these intersections is monotonically decreasing in $\ell$. We then give an upper bound on these quantities, including only the Lee volume for which explicit formulae are known, dating back to 1982 \cite{astola1982lee}. All these results are new to the best of our knowledge. The consideration of the intersection of balls is interesting for several fundamental aspects of coding theory. Our intersection results determine the packing and covering properties of Lee metric codes which are important for decoding and performance analysis of these codes. As recently addressed in \cite{ott2023geometrical}, finding upper and lower bounds on the size of the intersection of two balls that can
be computed fast would be helpful for
considering covering properties of codes. We believe that our techniques could therefore be of interest in this context.

The outline of the article is as follows. We start with the preliminaries in Section \ref{sec:pre}, including the Lee weight, the Lee volume and some useful bounds or asymptotic relations. In the second part of the preliminaries we examine the ratios of certain Lee volumes and also study the intersections of Lee spheres. Section \ref{sec:numberof} begins with the supersaturation estimates, which we then use in the container algorithm. We also show how we can use this method to estimate $t$-error-correcting codes of a fixed cardinality. This issue has already been studied in \cite{gruica2022densities} for codes in the Hamming, rank and sum-rank metric. There, the counting problem was modelled in a bipartite graph satisfying certain regularity conditions and the number of codes was estimated by the number of isolated nodes in the graph. We  modify this method slightly for the Lee metric and present it again in Section \ref{sec:bipartite} Finally, we conclude this paper in Section \ref{sec:comp}.

\section{Preliminaries} \label{sec:pre}
Throughout this article let $\mathbb{N}$ be the set of positive integers, $\Zm :=  \{0,1,\dots,m-1\}$ denotes the ring of integers modulo $m$ where, by abuse of notation, each residue class $i+m\mathbb{Z}$ is denoted by its representative $0\leq i \leq m-1$. By $P$ we denote the set of primes and by $[n]$ the set of  positive integers up to $n$. We omit ``$p \in P$'' when writing $p \to +\infty$. 
For nonnegative quantities depending on $n$, we write $f(n) \lesssim g(n)$ if there is some positive constant $C$ such that $f(n) \leq Cg(n)$ for all sufficiently large $n$.
We write $[n]$ to denote the set of positive integers up to $n$.
Given a finite set $X$, we write $\binom{X}{t}$ for the collection of $t$-element subsets of $X$, and $\binom{X}{\leq t}$ for the family of subsets of $X$ with at most $t$ elements. 
The following identities, respectively upper bound, involving the binomial coefficient are well known 
and will be useful later on.
\begin{lemma}\label{lem:binomprelim} (cf. \cite[Chapter 5]{graham1989concrete}) 
    Let $n \in \mathbb{N}$, $k \geq 0$ be an integer. Then
    \begin{enumerate}
        \item $\binom{n}{k}= \frac{n}{k}\binom{n-1}{k-1}$,
        \item $\binom{n-1}{k}= \frac{n-k}{n}\binom{n}{k}$,
        \item $\binom{n}{k}= \frac{n-k+1}{k}\binom{n}{k-1}$,
        \item $ \left( \frac{n}{k}\right)^{k} \leq \binom{n}{k} \leq \left( \frac{en}{k}\right)^{k}.$
    \end{enumerate}
\end{lemma}

\subsection{Lee Metric Codes}

We briefly recall the Lee metric and the volume of its balls. 

\begin{definition} 
  For $x  \in\Zm$ we define the \emph{Lee weight} to be
  \begin{align*}
      \wL : \Zm \rightarrow \left\{0, \dots , \left\lfloor \frac{m}{2} \right\rfloor \right\}, \quad x \mapsto \min\{x, m-x\}.
  \end{align*}
  We denote the average Lee weight over $\Zm$ by $\theta_{m}$ and extend the weight additively to $\Zm^{n}$, so that the Lee weight for $x \in \Zm^{n}$ is 
  \begin{align*}
      \wL(x):= \sum_{i=1}^{n}\wL(x_{i}).
  \end{align*}
  We define the \emph{Lee distance} of $x,y \in\Zm^{n}$ as 
  \begin{align*}
      d^{\textnormal{L}}(x,y):= \wL (x-y).
  \end{align*}
  A \emph{(Lee metric) code} is a subset $\mC \subseteq \Zm^n$. When $|\mC| \geq 2$ its \emph{minimum Lee distance} is$$d^{\textnormal{L}}(\mC) := \min\{ d^{\textnormal{L}}(x,y) \colon x,y \in \mC, \, x \ne y\}.$$
  A code $\mC$ is \emph{linear} if it is a finitely generated submodule of $\Zm^{n}$. The smallest number of elements in $\mC$ which generate $\mC$ as a $\Zm$-module is called the \emph{rank} of $\mC$. Since we restrict to $\Zp^{n}$ in the linear case all linear codes we consider are \emph{free}, i.e., the size of any generating set of $\mC$ determines the rank of $\mC$.
  \end{definition}
  Note that for $m=2$ and $m=3$ the Hamming distance and the Lee distance are identical by definition.
\begin{definition}
        The \emph{Lee metric ball} in $\Zm^{n}$ of radius $0 \leq r \leq \lfloor m/2 \rfloor$ with center $c \in \Zm^{n}$ is $$\ballL{n,r,c} := \{ y \in \Zm^{n}: d^{\textnormal{L}}(x,y) \leq r \}.$$ Its volume is independent of its center and we denote it by $ \bbql{n,r}$. 
\end{definition}
\begin{definition}
    For any vector $z \in \Zm^{n}$ we define its  \emph{Lee composition} $\ell (z)$ as $$\ell  (z):=(\ell_{0}(z), \ell_{1}(z), \dots,  \ell_{\lfloor m/2\rfloor}(z)) ,$$ where 
$$\ell_{i}(z):= |\{j \in [n] :  \wL (z_{j})=i \} |. $$ 
Clearly, $\sum_{i=0}^{\lfloor m/2 \rfloor} \ell_i(z)  = n$ and $\sum_{i=0}^{\lfloor m/2\rfloor} i \ell_i(z) = w^\textnormal{L}(z).$
We denote by $\mathcal{L}_{m}^{n}:=\{\ell_{\mathbf{0}}, \dots, \ell_{\mathbf{\alpha}}\}= \{ \ell (z) : z \in \Zm^{n} \}$ the set of all Lee compositions, where $\alpha +1$ is the number of different Lee compositions in $\mathbb{Z}_m^n$ and $\ell_{\mathbf{0}}=\ell (0)=(n,0, \dots, 0).$ 
We then define the \emph{Lee association scheme} such that $R=(R_{\ell_{\mathbf{0}}},\dots, R_{\mathbf{\ell_{\alpha}}})$ forms a partition on $\Zm^{n} \times \Zm^{n}$ by  $$(x,y) \in R_{\mathbf{\ell_{i}}} \Leftrightarrow \ell (x-y)=\ell_{\mathbf{i}},$$
while satisfying the axioms $1.-4.$, see \cite[Section 2.1]{sole1991parameters}:
\begin{enumerate}
    \item $R_{\ell_{\mathbf{0}}}= \{(x,x) \in \Zm^{n} \times \Zm^{n} : x \in \Zm^{n}\}$.
    \item $R_{\ell_{\mathbf{i}}}^{-1} = \{(y,x) \in \Zm^{n} \times \Zm^{n} : (x,y) \in R_{\ell_{\mathbf{i}}}\} = R_{\ell_{\mathbf{j}}}$ for some $\mathbf{j} \in \{\mathbf{0}, \dots, \mathbf{\alpha}\}$.
    \item The size of $\{z \in \Zm^{n} : (x,z) \in R_{\ell_{\mathbf{i}}} \, \text{and} \, (z,y) \in R_{\ell_{\mathbf{j}}}\}$ is described by a number $p_{i,j}^{k}$ which depends on $k$, but not on $(x,y) \in R_{\ell_{\mathbf{k}}}.$
    \item For all $i,j,k \in \{0, \dots, \alpha\}$ we have $ p_{i,j}^{k}= p_{j,i}^{k}.$
\end{enumerate}
\end{definition}
\begin{remark}
    If $(x,y)\in R_{\ell_{\mathbf{i}}}$, then the Lee distance between $x$ and $y$ is 
$$d^{\textnormal{L}}(x,y) = \ell_{\mathbf{i},1}(x-y)+2\ell_{\mathbf{i},2}(x-y)+ \dots +\left\lfloor \frac{m}{2} \right\rfloor \ell_{\mathbf{i},\lfloor m/2 \rfloor}(x-y).$$
This directly implies that the Lee metric is constant on the classes of the scheme. Let the classes of pairs $(x,y)$ with  $d^{\textnormal{L}}(x,y)=j$ be denoted by $R_{\ell_{\mathbf{j_1}}}, \dots , R_{\ell_{\mathbf{j_s(j)}}},$ where $\mathbf{s(j)}$ is the number of different compositions having weight $j$.
This number $\mathbf{s(j)}$ corresponds to the number of weak compositions of $j$ into $n$ parts such that each part is at most $\lfloor m/2 \rfloor.$ 
This number can be computed as
$$ \mathbf{s(j)} = \sum_{\substack{(r,s) \in \mathbb{N}^{2} \\ r+s\lfloor m/2 \rfloor=j}}(-1)^s \binom{n}{s} \binom{n+r-1}{r},$$
see, e.g., \cite[Chapter 1]{stanley2011enumerative}.
Thus, 
$$\mid \mathcal{L}_m^n\mid = \alpha+1 =\sum_{j=0}^{n \lfloor m/2 \rfloor} \mathbf{s(j)} .$$
\end{remark}

By axiom $3.$ we obtain: 
\begin{lemma} \label{lem:scheme}
    Let $(x,y) \in \Zm^{n} \times \Zm^{n}$ such that $d^{\textnormal{L}}(x,y)=k$. Then the number of vectors $z \in \Zm^{n}$ with $d^{\textnormal{L}}(x,z) \leq r$ and $d^{\textnormal{L}}(y,z) \leq r$ is equal to 
    $$\sum_{i=0}^{r}\sum_{j=1}^{\mathbf{s(i)}}\sum_{\ell=1}^{\mathbf{s(i)}}p_{i_{j},i_{\ell}}^{k},$$
    where $p_{i,j}^{k}:=|\{z \in \Zm^{n}: (x,z) \in R_{\ell_{\mathbf{i}}} \wedge (y,z) \in R_{\ell_{\mathbf{j}}}\}|$ which depends on $k$, but not on the pair $(x,y) \in R_{\ell_{\mathbf{k}}}$.
\end{lemma}
\begin{definition}
    A \emph{$t$-Lee-error-correcting code} of length $n$ is a code in $\Zm^{n}$ with minimum Lee distance at least $2t+1$. Its maximum size is denoted by $A_{m}(n, 2t+1)$.
\end{definition}

A simple counting argument shows the following:
\begin{theorem}
Let $A_{m}(n, 2t+1)$ denote the maximum size of a $t$-Lee-error-correcting code $\mC \subseteq \Zm^{n}$ and 
\begin{equation} \label{eq:Hamming}
    H_{m}(n,t):=\frac{m^{n}}{\bbql{n,t}},
\end{equation}
then 
\begin{align*}
    A_{m}(n, 2t+1) \leq H_{m}(n,t).
\end{align*}
\end{theorem}
  This upper bound is known as the \emph{sphere packing (or Hamming) bound}. Codes that achieve equality for this upper bound are called \emph{perfect codes}.

 Furthermore, many Singleton-like bounds are known for Lee metric codes (see \cite{byrne2023bounds} for an overview on these bounds). 
 However, the smallest known bound on the minimum Lee distance of  free linear codes over $\Zps$ can be achieved using a Plotkin-like argument, where one averages the weight over $\Zps$. 
\begin{theorem}(Plotkin-like bound, see \cite[Theorem 26]{byrne2023bounds}) \label{thm:Plotkin}
Let $\mC \subseteq \Zps^n$ be a free linear $t$-Lee-error-correcting code. For $p \geq 3$, we have
\begin{align*} 
    |\mC| \leq p^{s(n-\frac{8t+4}{p^s+1}+1)}.
\end{align*}
\end{theorem}

\subsection{Intersections and Estimates of the Lee Metric Volume}

The volumes of the Lee metric balls are known:
\begin{lemma}\cite[Chapter 6.3]{astola2015algebraic}\label{lem:Leevolumes}
      \begin{align*} 
    \bbql{n,r} = \begin{cases*}
\displaystyle\sum_{i=0}^{\lfloor 2r/m\rfloor}(-1)^{i}\binom{n}{i}\mathlarger{\mathlarger{\sum}}_{j=0}^{r-(m/2)i}2^{j}\binom{n}{j}\binom{r-(m/2)i}{j} &  \textnormal{if $m$ is even},\\
\displaystyle\sum_{i=0}^{r}\binom{n+1}{i}\displaystyle\sum_{j=0}^{\lfloor 2r/m\rfloor}(-2)^{j}\binom{n}{j}\binom{n-j}{r-j(\lfloor m/2\rfloor+1)-i} &  \textnormal{if $m$ is odd}. \\
\end{cases*}
\end{align*}
\end{lemma}

This gives rise to the following bounds:
\begin{lemma}\label{lem:leebounds}
    Let $n \in \mathbb{N}$ and $1 \leq r \leq n\lfloor \frac{m}{2} \rfloor$. Then 
    \begin{enumerate} 
        \item  $\bbql{n,r} \leq \mathlarger{\mathlarger{\sum}}_{i=0}^{n}2^{i}\binom{n}{i}\binom{r}{i},$
        \item $\bbql{n,r} \geq \begin{cases*}
\binom{n}{r}2^{r} &  \textnormal{if $r <n$,} \\
2^{n} &  \textnormal{if $r \geq n$.}
\end{cases*}$

    \end{enumerate}

\end{lemma}

\begin{proof}
    The upper bound is given in \cite[Section 10]{roth2006introduction}. The lower bound is known from \cite{weger2020np-complete}. Indeed, if $r <n$, $\bbql{n,r}$ contains all vectors in $\Zm^{n}$ with $r$ entries in $\{1, m-1\}$ and zero elsewhere. If $r \geq n$, $\bbql{n,r}$ contains all vectors in $\{1, m-1\}^{n}$. 
\end{proof}

We obtain the following asymptotic behavior of the volumes:
\begin{lemma} \label{lem:volumeLandau}
    Let $n \in \mathbb{N}, C>0$ be a constant and $1 \leq r \leq n\lfloor \frac{m}{2} \rfloor$. Then 
    \begin{enumerate} 
        \item $\bbql{n,r} \in \omega(1)$ as $n \to + \infty.$
        \item $\bbql{n,r} \in O(1)$ as $m \to + \infty.$
        \item If $r \lesssim n^{a/b}$, where $a <b$, then $\bbql{n,r} \in o\left(\frac{C^{n}}{n^{(b+1)/b}}\right)$ as $n \to + \infty.$
    \end{enumerate}  
    \begin{proof}
        The first two statements follow directly from Lemma \ref{lem:leebounds}. For $r \leq n^{a/b}$ we have
        $$\bbql{n,r} \leq \mathlarger{\mathlarger{\sum}}_{i=0}^{n^{a/b}}2^{i}\binom{n}{i}\binom{n^{a/b}}{i} \sim (2e^{2}n^{a/b})^{n^{a/b}}$$
        for $n\rightarrow +\infty$ (from Lemmata  \ref{lem:binomprelim} and \ref{lem:leebounds}). 
        Moreover
        \begin{align*}
            \frac{n^{(b+1)/b}(2e^{2}n^{a/b})^{n^{a/b}}}{C^n}= \frac{\text{exp}\left( \frac{b+1}{b}\log (n)+ n^{a/b}(\log(2)+2+\log(n^{a/b}))   \right)}{\text{exp}(n\log(C))} \stackrel{n\rightarrow +\infty}{\rightarrow} 0,
        \end{align*}
        since the exponent of the denominator grows linearly in $n$, while the exponent of the nominator only grows by $n^{a/b}.$
    \end{proof}
    \end{lemma}

Before studying the intersection of two Lee balls we derive estimates of the volumes of balls of different radii. The proofs can be found in the appendix.

\begin{lemma} \label{lem:volumeestimate}
    Let $1 \leq t \leq n\lfloor \frac{m}{2} \rfloor$ and $C_{m}$ be a sufficiently large constant. Then for any $1 \leq i$
    $$ \bbql{n, t} \leq \begin{cases*}
 \left( \frac{t+i}{2(n+i)}\right)^{i}\bbql{n+i,t+i} &  \textnormal{if $m$ is even,}\\
 \left( \frac{C_{m}}{n+i}\right)^{i}\bbql{n+i,t+i} &  \textnormal{if $m$ is odd and $t=\lceil \frac{m}{2}\rceil$,} \\
 \left( \frac{t+i}{n+i}\right)^{i}\bbql{n+i,t+i} &  \textnormal{else,}
\end{cases*} $$
and for $1\leq i \leq n\lfloor \frac{m}{2} \rfloor-t$
    $$ \bbql{n, t} \leq \begin{cases*}
 \left( \frac{t+i}{2n}\right)^{i} \left( \frac{n-i +1}{n-i +1-t}\right)^{i}\bbql{n,t+i} &  \textnormal{if $m$ is even,}\\
 \left( \frac{C_{m}}{n}\right)^{i}\left(\frac{n-i+1}{n-i+1-t} \right)^{i}
\bbql{n,t+i} &  \textnormal{if $m$ is odd and $t=\lceil \frac{m}{2}\rceil$,} \\
\left( \frac{t+i}{n}\right)^{i}\left(\frac{n-i+1}{n-i+1-t} \right)^{i}
\bbql{n,t+i} &  \textnormal{else.} \\
\end{cases*} $$
\end{lemma}

Next we investigate the intersection of two balls with equal Lee radii. 
\begin{lemma} \label{lem:assoclee}
    Let $0 \leq t \leq n\lfloor \frac{m}{2} \rfloor$ and $c_1, c_2 \in \Zm^{n}$, then the size of $\ballL{n,r,c_1} \cap \ballL{n,r,c_2}$ depends on $c_1, c_2$ only through $d^{\textnormal{L}}(c_{1},c_{2}).$
 \end{lemma}
   \begin{proof}
This follows directly from Lemma \ref{lem:scheme}.
    \end{proof} 
Hence for $0 \leq t \leq n\lfloor m/2 \rfloor$, we denote the cardinality of the intersection of two balls in $\Zm^{n}$ of radius $t$ with centers being Lee distance $\ell$ apart as $W(m, n,t,\ell)$. When there is no ambiguity about the ambient space considered we simply denote it as $W(t,\ell)$. 
Obviously $W(t,\ell)=0$ when $\ell > 2t$.  Thus, we will only focus on $1 \leq \ell \leq 2t \leq nm,$ as $t \leq n \lfloor m/2 \rfloor.$

In the following we show that the numbers $W(t,\ell)$ are decreasing for increasing $\ell$.
\begin{lemma} \label{lem:mon}
   Let $1 \leq \ell \leq \text{min}\{ 2t, n\lfloor m/2\rfloor -1 \},$ then $$W(t,\ell+1) \leq W(t,\ell).$$
   \begin{proof}
   We define the notation
   $$(1^{\ell}0^{n-\ell}):= (\underbrace{1,1,\dots,1}_{\ell},\underbrace{0,0,\dots,0}_{n-\ell}) .$$
      By Lemma \ref{lem:assoclee} it is sufficient to determine the intersection of the radius $t$ Lee balls centered at $c_1 = (0^{n})$ and $ c_2= (1^{\ell}0^{n-\ell})$ for $W(t,\ell)$, respectively the balls centered at $c_1$ and $c_3= (1^{\ell+1}0^{n-\ell-1})$ for $W(t,\ell+1)$. 
      Define
       $$X:= \ballL{n,t,c_1} \cap \ballL{n,t,c_3}\setminus \ballL{n,t,c_2} \quad \textnormal{and} \quad Y:=  \ballL{n,t,c_1} \cap \ballL{n,t,c_2}\setminus \ballL{n,t,c_3}.$$ Let $x \in X$. Since $x \in \ballL{n,t,c_{1}} \cap \ballL{n,t,c_{3}}$ we have $d^{\textnormal{L}}(c_{1},x) \leq t$ and $d^{\textnormal{L}}(c_{3},x) \leq t$. As $x \notin \ballL{n,t,c_{2}}$ we also have $d^{\textnormal{L}}(c_{2},x) >t$. We obtain 
              \begin{align} \label{eq:smaller}
              w^{\textnormal{L}}(x_{1}-1, \dots , x_{\ell+1}-1, x_{\ell+2}, \dots , x_{n})\leq t\\\label{eq:larger}
              w^{\textnormal{L}}(x_{1}-1, \dots , x_{\ell}-1, x_{\ell+1}, \dots , x_{n})> t.
        \end{align}
       As $ (x_{1}-1, \dots , x_{\ell+1}-1, x_{\ell+2}, \dots , x_{n})$ and $(x_{1}-1, \dots , x_{\ell}-1, x_{\ell+1}, \dots ,x_{n})$ differ only in the $(\ell+1)^{th}$ coordinate by exactly one, we have from \eqref{eq:smaller} and \eqref{eq:larger} that \begin{align}
           \label{eq:subs1}
            w^{\textnormal{L}}(x_{1}-1, \dots , x_{\ell+1}-1, x_{\ell+2}, \dots , x_{n})&=t\\
           w^{\textnormal{L}}(x_{1}-1, \dots , x_{\ell}-1, x_{\ell+1}, \dots ,x_{n})&=t+1. \label{eq:subs2}
       \end{align} 
       Hence, $d^{\textnormal{L}}(x,c_{2}) = t+1$ and $d^{\textnormal{L}}(x,c_{3})=t.$ Moreover, subtracting \eqref{eq:subs1} from \eqref{eq:subs2} gives
       $w^{\textnormal{L}}(x_{\ell+1}) - w^{\textnormal{L}}(x_{\ell+1}-1) =1$, from which we deduce that $x_{\ell+1}\in \{1, \dots ,\lfloor m/2 \rfloor\},$ as otherwise $w^{\textnormal{L}}(x_{\ell+1})-w^{\textnormal{L}}(x_{\ell+1}-1) = -1$ and in particular, $d^\textnormal{L}(x,c_2) < d^\textnormal{L}(x,c_3),$ which is not possible. 
       
       To show that $|X| \leq |Y|$ we define an injective map $\phi: X \rightarrow Y$ by $$\phi(x)_{i} := \begin{cases}
           x_i & \text{ if } \, i \in [n]\setminus \{\ell+1\}, \\
           m-x_{i}+1 &  \text{if } \, i = \ell+1.
       \end{cases}$$ 
       In the following we show that this map is well-defined. For this we will show that for every $x \in X$ we have \begin{align*}
           d^{\textnormal{L}}(\phi (x), c_{1}) <d^{\textnormal{L}}(x,c_1) \leq t,\quad d^{\textnormal{L}}(\phi (x), c_{3}) =t+1, \quad  d^{\textnormal{L}}(\phi (x), c_{2})=t.
       \end{align*}
       The first observation follows from 
       the fact that $x_{\ell+1}\in \{1, \dots ,\lfloor m/2 \rfloor\}$ and hence
       $$w^{\textnormal{L}}(\phi (x)_{\ell+1}) = x_{\ell+1}-1 < x_{\ell+1} = w^{\textnormal{L}}(x_{\ell+1}).$$
       For the second observation recall that $d^{\textnormal{L}}(c_{3},x) =t$. 
       We can deduce that
       \begin{align*}
           w^{\textnormal{L}}(\phi(x)_{\ell+1}-1) = w^{\textnormal{L}}(m-x_{\ell+1})=x_{\ell+1}&= (x_{\ell+1}-1)+1= w^{\textnormal{L}}(x_{\ell+1}-1)+1, 
       \end{align*}
       which gives $d^{\textnormal{L}}(\phi(x),c_{3})=t+1.$ The last observation follows similarly
       and hence $\phi(x) \in Y$. As no two distinct elements are mapped onto the same image we have $|X| \leq |Y|$ and thus $$| \ballL{n,t,c_1} \cap \ballL{n,t,c_3}| = |X \sqcup Z| \leq  |Y \sqcup Z| =| \ballL{n,t,c_1} \cap \ballL{n,t,c_2}|$$
       where $Z=\ballL{n,t,c_1} \cap\ballL{n,t,c_2} \cap \ballL{n,t,c_3}$.
   \end{proof}
\end{lemma}

In order to obtain relations between the numbers $W(t,\ell)$ (which will be important for our saturation estimates given in Section \ref{sec:numberof}) we start by giving an explicit formula for $W(t,1)$ and an upper bound for $W(t,\ell)$ for general $\ell \in \{1, \ldots, n \lfloor m/2 \rfloor \}$.
\begin{lemma} \label{lem:interesti}
    Let $1 \leq \ell \leq \text{min}\{ 2t, n\lfloor m/2\rfloor -1 \}.$
    \begin{enumerate}
        \item For $\ell =1$ we have $$
            W(t,1) = \begin{cases}
           2 \displaystyle\sum_{i=1}^{m/2}\bbql{n-1,t-i} & \text{if $m$ is even},\\
        2 \displaystyle\sum_{i=1}^{(m-1)/2}\bbql{n-1,t-i} +\bbql{n-1,t-\frac{m-1}{2} } & \text{if $m$ is odd.}
       \end{cases}$$
         \item 
         For $\ell \geq 1$ we have $W(t,\ell) \leq m^{\ell} \bbql{n-\lceil \ell /2\rceil,t-\lceil \ell /2\rceil}.$
    \end{enumerate}
\begin{proof}
     Let $\ell =1$, suppose that the two balls are centered at $c_{1}:=(0^{n})$ and $c_{2}:=(10^{n-1})$ and let $x \in \ballL{n,t,c_{1}} \cap \ballL{n,t,c_{2}}$. This gives $w^{\textnormal{L}}(x_{1}) \leq t$ and because of translation invariance the remaining $n-1$ coordinates have Lee weight at most $$t-\max\{w^{\textnormal{L}}(x_1), w^{\textnormal{L}}(x_{1}-1)\}.$$ Observe that 
     \begin{align*}
       \max\{w^{\textnormal{L}}(x_1),w^{\textnormal{L}}(x_{1}-1)\} =\begin{cases}
          1 & \text{if} \, x_1 = 0, \\
           x_{1} & \text{if} \, 1\leq x_1 \leq \lfloor m/2\rfloor, \\
            m-x_{1} & \text{if} \, x_1 = \lfloor m/2\rfloor+1 \, \textnormal{and $m$ is odd},\\
           m-x_{1}+1 & \text{in all other cases}.
       \end{cases}
     \end{align*}
Hence, for a fixed
     \begin{itemize}
     \item $x_1=0$ there are $\bbql{n-1,t-1}$ many choices for the remaining $n-1$ coordinates of $x$, 
     \item $x_1\in \{1,2,\dots,\lfloor m/2\rfloor\}$ there are $\bbql{n-1,t-x_1}$ many choices for the remaining $n-1$ coordinates, 
     \item $x_1\in \{\lfloor(m+1)/2\rfloor+1,\dots,m-1\}$ there are $\bbql{n-1,t-m+x_1-1}$ many choices for the remaining $n-1$ coordinates, 
    \item $x_1=(m+1)/2$ and $m$ odd there are $\bbql{n-1,t-(m-1)/2}$ many choices for the remaining $n-1$ coordinates.
     \end{itemize}
     For the second statement suppose that the two balls are centered at $c_{1}=(0^{n})$ and $c_{2}=(1^{\ell}0^{n-\ell})$. For $x \in \ballL{n,t,c_{1}} \cap \ballL{n,t,c_{2}}$ we have that the last $n-\ell$ coordinates have Lee weight at most $t-\max\{\sum_{i=1}^{\ell}w^{\textnormal{L}}(x_i), \sum_{i=1}^{\ell}w^{\textnormal{L}}(x_{i}-1)\}.$
     Moreover, we observe that 
     \begin{align} \label{eq:max_k}
         \max\left\{\sum_{i=1}^{\ell}w^{\textnormal{L}}(x_i), \sum_{i=1}^{\ell}w^{\textnormal{L}}(x_{i}-1)\right\} \geq \left\lceil \ell/2 \right\rceil.
     \end{align}
    Let $s$ be the support size of the first $\ell$ coordinates of $x$. If $s \geq \lceil \ell/2 \rceil$ \eqref{eq:max_k} can be immediately derived from the fact that the Lee weight of the first $\ell$ coordinates is at least as large as the support size. If $s < \lceil \ell/2 \rceil$ at least $\lceil \ell/2 \rceil$ coordinates of $x$ are zero and therefore $\sum_{i=1}^{\ell}w^{\textnormal{L}}(x_{i}-1)  \geq \lceil \ell/2 \rceil$.
     Now the upper bound can be derived using the monotonicity of the Lee volume in both parameters. 
\end{proof}
\end{lemma}

Using Lemma \ref{lem:interesti} we can see that the intersection of two balls with centers being distance $\ell$ apart is significantly smaller than the intersection of two balls with centers being distance $1$ apart. Moreover, by applying Lemma \ref{lem:volumeestimate} we can give an upper bound on $W(t,1)$ which then contains only the volume $\bbql{n,t}$. 
\begin{lemma} \label{lem:estimateinter}
    Let $0 \leq \ell < \min\{t, n \lfloor m/4 \rfloor\},$ and $C_{m}$ be a sufficiently large constant, then
\begin{align*}
  W(t,2\ell+1) &
    \leq \begin{cases}
          \left( \frac{mt}{2n}\right)^{\ell} W(t,1) \leq \left( \frac{mt}{2n}\right)^{\ell+1} \bbql{n,t} & \text{if } m \text{ is even,} \\
           \left(\frac{mC_{m}}{n}\right)^{\ell}W(t,1)\leq\left(\frac{mC_{m}}{n}\right)^{\ell+1} \bbql{n,t}  &
     \text{if } m \text{ is odd and } t=\lceil \frac{m}{2}\rceil,
   \\  \left( \frac{mt}{n}\right)^{\ell} W(t,1) \leq
             \left( \frac{mt}{n}\right)^{\ell+1}\bbql{n,t} & \text{in all other cases.}
    \end{cases}
\end{align*}
    \begin{proof}
     We only show the case where $m$ is even, the other cases follow from similar arguments. 
     
     Combining both statements from Lemma \ref{lem:interesti}, together with Lemma \ref{lem:volumeestimate} we can deduce the first statement, i.e.,
     \begin{align*}
              \frac{W(t,2\ell+1)}{W(t,1)} \leq \frac{m^{\ell}\bbql{n-\ell-1, t-\ell-1}}{\bbql{n-1,t-1}} \leq m^{\ell} \left( \frac{t}{2n}\right)^{\ell},
     \end{align*}
    which proves the first inequality. To prove the second inequality, apply again Lemma \ref{lem:volumeestimate} to yield
\begin{align*}
    \frac{W(t,1)}{\bbql{n,t}} \leq \frac{m\bbql{n-1, t-1}}{\bbql{n,t}} \leq  m \left(  \frac{t}{2n}\right).
\end{align*}
    \end{proof}
    \end{lemma}

\section{The Number of $t$-Lee-Error-Correcting Codes}\label{sec:numberof}

In this section we give an upper bound on the number of $t$-Lee-error-correcting codes, where we do not prescribe the cardinality of the codes. Throughout this section, unless otherwise stated, we allow $t$ to be a function in $n$. More precisely, we let $\lim\limits_{n \to +\infty} t(n)/(10 \sqrt{\theta_{m}n})=T,$ where $T \in [0,1]$ and $\theta_{m}$ denotes the average Lee weight over $\Zm$. However, we will often use $t$ as a shorthand for $t(n)$.

 To obtain an upper bound on the number of $t$-error-correcting codes we use a similar container algorithm as in \cite{balogh2016applications}. In order to apply this algorithm, we will first derive the necessary supersaturation estimates. 

Throughout this section $G=(V,E)$ is a simple, undirected graph without self-loops, nodes $V=\Zm^{n}$ and edges $E= \{(x,y) \in V^{2}: d^{\textnormal{L}}(x,y) \leq 2t\}.$ Any set $I \subseteq V$ that does not contain an edge is called an \emph{independent set} and corresponds to a $t$-Lee-error-correcting code. 
Moreover, we write $\Delta(G)$ for the \emph{maximum degree} of $G$ and $i(G)$ for the \emph{number of independent sets} in $G$. By $N_{G}(u)$ we denote the neighbourhood of a node $u$, i.e., the subset of nodes that are adjacent to $u$. Hence,
$$\Delta(G) = \max \{ \mid N_G(u) \mid : u \in V\}.$$
The container algorithm is based on a dynamic graph $(G_{i})_{i \geq 0}$ (initialized as $G_{0}=G$), iteratively reducing its node set $V(G_{i})$ by at least one node. More precisely, in each iteration of the algorithm the node with maximum degree (and possibly its neighbourhood) is removed from the graph.

Let $\mC \subseteq \Zm^{n}$ be a code and $G[\mC]=(\mC,E[\mC])$ be the corresponding induced subgraph of $G$.
For each $1 \leq r \leq 2t$, we define
$$E_{r} := \{(x,y) \in E[\mC] : d^{\textnormal{L}}(x,y) =r\}$$
and for $v \in \mC$ the number of edges in $E_{r}$ incident to $v$ by
$$\deg_{r}(v) := |\{u \in \mC: d^{\textnormal{L}}(u,v) =r\}|.$$ 

\begin{lemma}\label{lem:graphdegree}
Let $v \in \mC$ and $r\leq n/2$. Then
    $$\deg_{r}(v) \leq 3^{r}\binom{n}{r}$$ and therefore
    $$|E_{r}| \leq \frac{3^{r}}{2}\binom{n}{r}|\mC|.$$
\end{lemma}
\begin{proof}
With the binomial theorem we get $\deg_{r}(v)\leq \bbql{n,r} \leq \sum_{i=0}^{n}2^{i}\binom{n}{i}\binom{r}{i}=\sum_{i=0}^{r}2^{i}\binom{n}{i}\binom{r}{i}$
$\leq \binom{n}{r}\sum_{i=0}^{r}2^{i}\binom{r}{i} =3^r \binom{n}{r}$, which proves the first statement. The second follows from summing over all pairs of codewords in $\mC$.
\end{proof}
We have the following  estimates for the Lee metric. 
\begin{lemma} \label{lem:supersat01}
    Let $\mC \subseteq \Zm^n$ be a $t$-Lee-error-correcting code and $C_{m}$ be a sufficiently large constant. If $|\mC| \geq 2H_{m}(n,t)$, 
    then
    \begin{align*}
    \sum_{r=1}^{2t}W(t,r)|E_{r}| \geq \frac{|\mC|^{2}\bbql{n,t}^{2}}{10m^{n}}\quad \textnormal{and} \quad
\quad|E[\mC]| \geq \begin{cases}
\frac{n|\mC|^{2}}{5mtH_{m}(n,t)} & \,\textnormal{if $m$ is even}, \\
\frac{n|\mC|^{2}}{10mC_{m} H_{m}(n,t)} & \, \textnormal{if $m$ is odd and $t= \lceil \frac{m}{2}\rceil$},\\
\frac{n|\mC|^{2}}{10mt H_{m}(n,t)} & \, \textnormal{in all other cases}.\\
\end{cases}   \end{align*}
\begin{proof}
    Similar to \cite[Lemma 5.3]{balogh2016applications}, we use Lemma \ref{lem:mon} and Lemma \ref{lem:estimateinter} to obtain the result. For completeness we include the proof in the appendix. 
\end{proof}
\end{lemma}

The container method gives estimates for the number of independent sets of a graph by showing that the independent sets lie only in a small number of subsets of the node set which are `almost independent sets`. 
We will study the more detailed relationships after we have established the required saturation estimates. Our first supersaturation result gives a constraint on the size of the node set to ensure that the maximum degree of the induced graph is above a certain bound. 

\begin{proposition} (Supersaturation 1) \label{lem:supersat1}
    Let $ t\le  10 \sqrt{\theta_{m}n}$, $C_{m}$ be a sufficiently large constant and $\mC \subseteq \Zm^{n}$ be a $t$-Lee-error-correcting code. If $ |\mC| \ge  n^{4}H_{m}(n,t)$, then
    $$\Delta (G[\mC]) \gtrsim \begin{cases}
\frac{n^{3}|\mC|}{C_{m}H_{m}(n,t)} & \, \textnormal{if $m$ is odd and $t= \lceil \frac{m}{2}\rceil$},\\
\frac{n^{3}|\mC|}{t^{3}H_{m}(n,t)}  & \, \textnormal{in all other cases},\\
\end{cases}$$
asymptotically with respect to $n$.
\end{proposition}

        \begin{proof}
        We will show the statement for the case where $m$ is even, the proof is similar if $m$ is odd.
        Since $|\mC| \geq n^{4}H_{m}(n,t)$ we have from Lemma \ref{lem:supersat01}
$$|E[\mC]| \geq \frac{|\mC|^{2}n}{5mtH_{m}(n,t)} \geq |\mC|\frac{n^{5}}{5mt}.$$ 
Together with $|E_1|+|E_2| \lesssim n^{2}|\mC|$ and $|E_3|+|E_4| \lesssim n^{4}|\mC|$ (which follows from Lemma \ref{lem:graphdegree}) this gives $|E_1|+|E_2| \lesssim \frac{t|E[\mC]|}{n^{3}}$ and $|E_3|+|E_4| \lesssim \frac{t|E[\mC]|}{n}.$
Furthermore, by Lemma \ref{lem:estimateinter} we have
\begin{align*}
    &W(t,2) \leq W(t,1),\\
    &W(t,4) \leq W(t,3) \lesssim \frac{tW(t,1)}{n},\\
    & W(t,r) \leq W(t,5) \lesssim \frac{t^{2}W(t,1)}{n^{2}} \quad \forall r \geq 5,
\end{align*}
which gives
$$\displaystyle \sum_{r=1}^{2t}W(t,r) |E_{r}| \lesssim \left( \frac{t}{n^{3}}+\frac{t^{2}}{n^{2}} +\frac{t^{2}}{n^{2}}\right)|E[\mC]|W(t,1) \lesssim \frac{t^{2}}{n^{2}}W(t,1)|E[\mC]|.$$
By Lemma \ref{lem:supersat01} we have
 $$  \frac{|\mC|^{2}\bbql{n,t}^{2}}{m^{n}} \lesssim  \sum_{r=1}^{2t}W(t,r)|E_{r}|\lesssim \frac{t^{2}}{n^{2}} W(t,1)|E[\mC]|.$$
Using $W(t,1) \lesssim (t/n)\bbql{n,t}$ from Lemma \ref{lem:estimateinter} we obtain 
 $$|E[\mC]| \gtrsim \frac{n^{2}|\mC|^{2}\bbql{n,t}^{2}}{t^{2}m^{n}W(t,1)} \gtrsim \frac{n^{3}|\mC|^{2}}{t^{3}H_{m}(n,t)}.$$
 Thus, the average degree in $G[\mC]$ is $\gtrsim \frac{n^{3}|\mC|}{t^{3}H_{m}(n,t)}$, and hence $\Delta(G[\mC]) \gtrsim  \frac{n^{3}|\mC|}{t^{3}H_{m}(n,t)}.$
    \end{proof}
In the following we will give an estimate for the maximum degree that certifies that $i(G[\mC])$ is below $2^{(1+\varepsilon)H_{m}(n,t)}$. To get our estimate we consider the nodes of $G[\mC]$ in terms of their number of neighbours such that their Lee distance is at most $20$ from each other. By considering balls of radius $t$ and showing that their overlap is negligible we can estimate the number of $t$-Lee-error-correcting codes for $60 \le t \le  10 \sqrt{\theta_{m}n} $. Using a number less than $20$ would not give us suitable estimates in Lemma \ref{lem:C2}. Therefore, we choose this number. We define the sets 
\begin{align} \label{eq:S1}
    &\mC_{1}:=   \left\{v \in \mC : \deg_{r}(v) \leq \varepsilon n^{\lceil r/2\rceil/2} \quad \forall r \in \{1, \dots, 20\} \right\},\\ \label{eq:S2}
           &\mC_{2}:= \left\{v \in \mC : \displaystyle \sum_{r=1}^{20}\deg_{r}(v) \geq \frac{\log(n)}{\varepsilon} \right\}
\end{align}
and compute $i(G[\mC_{i}]),$ for $i=1,2$, assuming that the maximum degree of $G[\mC_{i}]$ is not too large. Note that for $v \in \mC$ we have $\sum_{r=1}^{20}\deg_{r}(v) \in o(n^{5})$ if $n \to +\infty$ and hence $\mC= \mC_{1} \cup \mC_{2}$ for sufficiently large $n$. Since every independent set of $G[\mC]$ can potentially be a union of an independent set of $G[\mC_{1}]$ and $G[\mC_{2}]$ when $n$ is sufficiently large, we get an overall estimate using $i(G[\mC]) \leq i(G[\mC_{1}])i(G[\mC_{2}])$.
\begin{lemma} \label{lem:C1}
    Let $t \leq  10 \sqrt{\theta_{m}n} $, $\mC \subseteq \Zm^{n}$ be a $t$-Lee-error-correcting code, $\Delta(G[\mC]) \leq n^{5}$ and fix $\varepsilon > 0$. Let 
 $\mC_{1}$ be as defined in \eqref{eq:S1}.

    Then $|\mC_{1}| \leq (1+O(\varepsilon))H_{m}(n,t)$, for sufficiently large $n$.
    \begin{proof}
      The proof can be found in the appendix.
    \end{proof}
\end{lemma}
\begin{lemma} \label{lem:C2}
    Let $60 \leq t \leq  10 \sqrt{\theta_{m}n} $, $\mC \subseteq \Zm^{n}$ be a $t$-Lee-error-correcting code and $\Delta(G[\mC]) \leq n^{5}$. Fix $0 <\varepsilon$, let $n$ be sufficiently large and $\mC_{2}$ be as defined in \eqref{eq:S2}.
    Then the number of independent subsets of $\mC_{2}$ satisfies
$$i(G[\mC_{2}]) \leq 2^{O(\varepsilon H_{m}(n,t))}.$$
    \begin{proof}
The proof can be found in the appendix.
      \end{proof}
\end{lemma}
\begin{proposition} (Supersaturation 2) \label{lem:supersat2}
    Let $0 < \varepsilon$, $60 \leq t \leq  10 \sqrt{\theta_{m}n} , \mC \subseteq \Zm^{n}$ be a $t$-Lee-error-correcting code and let $n$ be sufficiently large. If $\Delta (G[\mC]) \leq n^{5}$, then 
    \begin{align*}
        i(G[\mC]) \leq 2^{(1+\varepsilon)H_{m}
        (n,t)}.
    \end{align*}
    \begin{proof}
Defining $\mC_1$ and $\mC_2$ as in \eqref{eq:S1} and \eqref{eq:S2} gives
\begin{align*}
    i(G[\mC]) \leq i(G[\mC_{1}])i(G[\mC_{2}]) \leq 2^{|\mC_{1}|}i(G[\mC_{2}]) \leq 2^{(1+O(\varepsilon))H_{m}(n,t)}.
\end{align*}
    \end{proof}
\end{proposition} 
To proof the container lemma also for small $t$, i.e., $1 \leq t \leq 60$, we use a saturation estimate from \cite{balogh2016applications}.
\begin{lemma} \label{lem:supersat2b}
Let $\mC \subseteq \Zm^{n}$ and $C_{m}$ be a sufficiently large constant. If $|\mC| \geq H_{m}(n,t)+ \gamma$ then there are at least $$\begin{cases}
\frac{\gamma 2n}{mt} & \,\textnormal{if $m$ is even}, \\
\frac{\gamma n}{m C_{m}} & \, \textnormal{if $m$ is odd and $t= \lceil \frac{m}{2}\rceil$},\\
\frac{\gamma n}{mt}& \, \textnormal{in all other cases},\\
\end{cases}$$ pairs $v,w \in \mC$ that have Lee distance at most $2t$ from each other.
\begin{proof}
     Using Lemma \ref{lem:estimateinter} the proof is similar to the proof of  \cite[Lemma 5.2]{balogh2016applications} and can be found in the appendix.
\end{proof}
\end{lemma}
Assuming these saturation results, we can use a modification of the container lemma \cite[Lemma 2.3]{dong2022number}, which directly implies the main theorem.  For completeness, we give the proof. However, we will first describe the container algorithm and its properties.
\begin{lemma} \label{lem:Container}
    Let $1 \leq t \leq  10 \sqrt{\theta_{m}n} , 0< \varepsilon$ and $n$ be sufficiently large. Then there is a collection $\mathcal{F}$ of subsets of $\Zm^{n}$ with the following properties 
    \begin{enumerate}
    \item [(i)] $|\mF| \leq 2^{\varepsilon H_{m}(n,t)}.$
    \item [(ii)] Every $t$-Lee-error-correcting code $\mC \subseteq \Zm^{n}$ is contained in some $F \in \mF$.
    \item [(iii)] $i(G[F]) \leq 2^{(1+\varepsilon)H_{m}(n,t)}$ for every $F \in \mF$.
\end{enumerate} 
\end{lemma}
Recall that $G=(V,E)$ describes a graph with $V= \Zm^{n}$ and $E= \{(x,y) \in V^{2}: d^{\textnormal{L}}(x,y) \leq 2t\}$.
For the proof of Lemma \ref{lem:Container} we will apply the following version of the algorithm of Kleitman and Winston \cite{kleitman1980asymptotic, kleitman1982number}.

\begin{algorithm} \label{alg:1}  Fix an arbitrary order $v_{1}, \dots , v_{m^{n}}$ of the elements of $V$. Let $I$ be an independent set in $G$. Initialize $G_{0}:= G$ and $P:= \emptyset$. If $m$ is even set $\Delta := \varepsilon n /(mt)$, if $m$ is odd and $t=\lceil m/2 \rceil$ chose $C_{m}$ sufficiently large such that Lemma \ref{lem:supersat2b} is satisfied and set $\Delta := \varepsilon n /(m C_{m})$. In all other cases set $\Delta := \varepsilon n /(2mt)$. In each step $i$ of the algorithm add nodes to $P$ through the following process.
\begin{enumerate}
    \item[i)] Let $u$ be the node of maximum degree in $V(G_{i-1}).$ In case there are two nodes having maximum degree, choose the first node in the ordering. 
    \item[ii)] If $u \notin I$, define $V(G_{i}):= V(G_{i-1})\setminus \{u\}$ and proceed to step $i+1$.
    \item[iii)] If $u \in I$ and $\deg_{G_{i-1}}(u) \geq \Delta$ then add $u$ to $P$, define $V(G_{i}):= V(G_{i-1})\setminus(\{u\} \cup N_{G}(u))$ and proceed to step $i+1$.
    \item[iv)] If $u \in I$ and $\deg_{G_{i-1}}(u) < \Delta$ define $f(P):= V(G_{i})$, 
    terminate and output $P$ and $f(P)$. 
\end{enumerate}
\end{algorithm}
\begin{lemma} \label{lem:func1}
    Let $0< \varepsilon, 1 \leq t < 60$ and let $n$ be sufficiently large. Then there exists a function 
    $$f: \binom{V(G)}{\leq \frac{\varepsilon H_{m}(n,t)\log(n)}{n}} \rightarrow \binom{V(G)}{\leq \left(1+\varepsilon \right)H_{m}(n,t)},$$
    such that for every independent set $I$ in $G$ there is some subset $P \subseteq I$ that satisfies
\begin{enumerate}
    \item[(i)] $|P| \leq \frac{\varepsilon H_{m}(n,t)\log(n)}{n}$,
    \item[(ii)] $ I \subseteq P \cup f(P)$.
     \item[iii)] $i(G[P \cup f(P)]) \leq 2^{(1+\varepsilon)H_{m}(n,t)}$.
\end{enumerate}
\begin{proof}
    The proof can be found in the appendix.
\end{proof}

    \end{lemma}
       To show the statement for $60 \leq t \leq  10 \sqrt{\theta_{m}n} $, we use the following variant of Algorithm \ref{alg:1}, proposed by Dong et al. \cite{dong2022number}.
    \begin{algorithm}\label{alg:2} Fix an arbitrary order $v_{1}, \dots , v_{m^{n}}$ of the elements of $V$. Let $I$ be an independent set in $G$. Initialize $G_{0}:= G$ and $P:= \emptyset$. In each step $i$ of the algorithm add nodes to $P$ through the following process.
\begin{enumerate}
\item[i)] If $i(G_{i-1}) \leq 2^{(1+\varepsilon)H_{m}(n,t)},$ set $f(P)=V(G_{i-1})$, terminate and output $P$ and $f(P)$.
    \item[ii)] Let $u$ be the node of maximum degree in $V(G_{i-1}).$ In case there are two nodes having maximum degree, choose the first node in the ordering. 
    \item[iii)] If $u \notin I$, define $V(G_{i}):= V(G_{i-1})\setminus \{u\}$ and proceed to step $i+1$.
    \item[iv)] If $u \in I$, then add $u$ to $P$, define $V(G_{i}):= V(G_{i-1})\setminus(\{u\} \cup N_{G}(u))$ and proceed to step $i+1$.
\end{enumerate}
\end{algorithm}
 Algorithm 2 certifies that the following function is well defined.
\begin{lemma} \label{lem:func2}
     Let $0<\varepsilon, 0<C$ be a sufficiently large constant, $60 \leq t \leq  10 \sqrt{\theta_{m}n} $ and let $n$ be sufficiently large. Then there exists a function 
    $$f: \binom{V}{\leq \frac{CH_{m}(n,t)\log(n)}{n}} \rightarrow 2^{V},$$
    such that for any independent set $I$ in $G$ there is some subset $P \subseteq I$ that satisfies the following three conditions:
    \begin{enumerate}
        \item[i)] $ I \subseteq P \cup f(P)$,
         \item[ii)] $|P| \leq \frac{CH_{m}(n,t)\log(n)}{n}$,
          \item[iii)] $i(G[P \cup f(P)]) \leq 2^{(1+\varepsilon)H_{m}(n,t)}$.
    \end{enumerate}
   \begin{proof}
       The first observation follows by structural induction. To prove the remaining two conditions, we distinguish two stages of the algorithm according to the size of $V(G_{i})$.
    \begin{itemize}
        \item Let $P_{1}$ denote the set of nodes $u \in P$ added to $P$ when $|V(G_{i-1})| \geq n^{4}H_{m}(n,t) $.
           \item Let $P_{2}$ denote the set of nodes $u \in P$ added to $P$ when $|V(G_{i-1})| < n^{4}H_{m}(n,t).$
    \end{itemize}
That is, there exists a $k >0$ such that every node added to $P$ up to and including step $k$ is in $P_{1}$, and every node added to $P$ after step $k$ is in $P_{2}$. Let $\beta := \frac{n^{3/2}}{H_{m}(n,t)}$.
     As we add nodes to $P_1$, by Lemma \ref{lem:supersat1} we remove at least a $\beta$ fraction of nodes from $G_{i-1}$ to get $G_{i}$. 
 So $(1-\beta)^{k}m^{n} \leq n^{4}H_{m}(n,t)$, respectively $(1-\beta)^{k}\bbql{n,t} \leq n^{4}$. 
    By Lemma \ref{lem:volumeLandau} we have $\beta \rightarrow 0$ as $n \to +\infty$, so that $\beta \leq \log \left( \frac{1}{1-\beta}\right)$ for sufficiently large $n$ and therefore 
    \begin{align*}
        |P_{1}| &\leq \frac{\log\left( \frac{n^{4}H_{m}(n,t)}{m^{n}}\right)}{\log \left( 1-\beta\right)} \leq \frac{\log\left( \frac{m^{n}}{n^{4}H_{m}(n,t)}\right)}{\log \left( \frac{1}{1-\beta}\right)} \lesssim \frac{\log(\bbql{n,t})-4\log(n)}{\beta}\stackrel{(*)}{\lesssim}  \frac{t\log(n)H_{m}(n,t)}{n^{3/2}}\\
        &\lesssim \frac{\log(n)H_{m}(n,t)}{n},
    \end{align*}
where (*) follows from Lemmata \ref{lem:binomprelim} and \ref{lem:leebounds}.
    As we add nodes to $P_2$ we have $i(G_{i-1})\geq 2^{(1+\varepsilon)H_{m}(n,t)}$ during the associated iteration. By Lemma \ref{lem:supersat2} we remove at least $n^{5}$ nodes in each iteration and hence
    $$|P_{2}| \lesssim \frac{n^4H_{m}(n,t)}{n^5}\leq \frac{H_{m}(n,t)}{n}.$$

This gives $$|P|= |P_1|+|P_2| \leq \frac{C \log(n)\cdot H_{m}(n,t)}{n},$$
 and thus 
 $$i(G[P \cup f(P)]) \leq 2^{|P|} \cdot i(f(P)) \leq 2^{(1+2\varepsilon)H_{m}(n,t)}.$$
   \end{proof}
\end{lemma}

 We can now proceed with the proof of Lemma \ref{lem:Container}.
 \begin{proof}[Proof of Lemma \ref{lem:Container}]
 Let $G=(V,E)$ be the graph with $V= \Zm^{n}$ and $E= \{(x,y) \in V^{2}: d^{\textnormal{L}}(x,y) \leq 2t\}$.
Suppose that $I$ ranges over all independent sets of $G$. If $1 \leq t \leq 60$, let $\mF$ be the collection of all $P \cup f(P)$ obtained by Algorithm \ref{alg:1}, otherwise let $\mF$ be the collection of all $P \cup f(P)$ obtained by Algorithm \ref{alg:2}. Since $t$-Lee-error-correcting codes are exactly the independent sets in $G$, we get from Lemmata \ref{lem:func1} and \ref{lem:func2} that every $t$-Lee-error-correcting code is contained in some $F \in \mF$. The third condition in Lemmata \ref{lem:func1} and \ref{lem:func2} gives $i(G[F]) \leq 2^{(1+2\varepsilon)H_{m}(n,t)}$, for each $F \in \mathcal{F}$. By Lemma \ref{lem:binomprelim} $\log \binom{n}{m} \leq m \log (O(n)/m)$. This gives for $1 \leq t <60$
  \begin{align*}
     \log |\mF| &\leq \log \binom{m^{n}}{\leq \frac{\varepsilon H_{m}(n,t)\log(n)}{n}} \leq \frac{\varepsilon H_{m}(n,t)\log(n)}{n}\log\left( O\left( \frac{\bbql{n,t}n}{\varepsilon \log(n)}\right)\right)\\ &\lesssim \frac{\varepsilon H_{m}(n,t) \log(n)}{n}\log\left( O\left( \frac{(2en)^{t}n}{t^{t}\varepsilon \log(n)}\right)\right) 
 \end{align*}
 and similarly, for $60 \leq t \leq  10 \sqrt{\theta_{m}n} $ and sufficiently large $C$ that

 \begin{align*}
     \log |\mF| \leq \log \binom{m^{n}}{\leq \frac{C \log(n)}{n}H_{m}(n,t)} &\lesssim \frac{C\log(n)}{n}H_{m}(n,t)\log\left( O\left( \frac{n^{t+1}(2e/t)^{t}}{C\log(n)}\right)\right)\\
     &\lesssim \frac{C t (\log(n))^{2}}{n}H_{m}(n,t).
 \end{align*}
Hence, $|\mF| \leq 2^{\varepsilon H_{m}(n,t)},$ for sufficiently large $n$. This finishes the proof of Lemma \ref{lem:Container}.
 \end{proof}
We can now state our main result on the number of $t$-Lee-error-correcting codes.
\begin{theorem}
    Let $1 \le n, 2 \le m, 1 \leq t \leq  10 \sqrt{\theta_{m}n} ,$ then the number of $t$-Lee-error-correcting codes in $\Zm^{n}$ is asymptotically at most $2^{(1+o(1))H_{m}(n,t)},$ for sufficiently large $n$.
    \begin{proof}
        Let $0 <\varepsilon$ and $\mF$ be as in Lemma \ref{lem:Container}, then the number of $t$-Lee-error-correcting codes in $\Zm^{n}$ is at most
        \begin{align*}
            \sum_{F \in \mF}i(G[F])\leq |\mF|2^{(1+\varepsilon)H_{m}(n,t)}=2^{(1+2\varepsilon)H_{m}(n,t)},
        \end{align*}
        since $\varepsilon$ can be chosen arbitrarily small, the result follows.
    \end{proof}
\end{theorem}
In the following we give an upper bound on the number of $t$-Lee-error-correcting codes of given size. We will use this result for experimental results in Python and compare it with the upper bound obtained from bipartite graphs in Section \ref{sec:bipartite}.

\begin{theorem}
      Let $ 2 \le m, 1 \leq t \leq 60,$ $0 < \varepsilon $ and let $1 \le n$ be sufficiently large, then the number of $t$-Lee-error-correcting codes of length $n$ and size $S$ in $\Zm^{n}$ is asymptotically at most
      $$2^{\varepsilon H_{m}(n,t)}\binom{(1+\varepsilon)H_{m}(n,t)}{S}.$$
      \begin{proof}
      For $1 \le t \le 60$ we let $\mF$ be the collection of all $P \cup f(P)$ obtained by Algorithm \ref{alg:1}. From Lemma \ref{lem:func1} we have 
$$    |P \cup f(P)| \leq \frac{\varepsilon H_{m}(n,t)\log(n)}{n}+\left(1+\varepsilon \right)H_{m}(n,t) \lesssim \left(1+\varepsilon \right)H_{m}(n,t),$$
i.e., an upper bound on the size of each container. Since any $t$-error-correcting code is contained in some $F \in \mF$ and $|\mF| \leq 2^{\varepsilon H_{m}(n,t)}$ the statement follows.
      \end{proof}
\end{theorem}
\begin{remark}
    Note that due to Lemma \ref{lem:volumeLandau} using supersaturation estimates depending on $H_{m}(n,t)$ when considering $m$ towards infinity, we would only consider the trivial code $\mC = \mathbb{Z}_{m}^{n}$ and end up with trivial counting results. The container method is therefore not suitable for considering the case of increasing $m$. In Section \ref{sec:bipartite} we will use a different counting technique, based on bipartite graphs, to also estimate the number of codes with a prescribed distance as $m$ tends to infinity. 
\end{remark}
Recall that $A_{m}(n, d)$ denotes the maximum size of a code $\mC \subseteq \Zm^{n}$ of minimum distance $d$.
We start by showing that $A_{m}(n, 2t+1) = o(H_{m}(n,t)/n)$ for $ 10 \sqrt{\theta_{m}n}  <t < n^{4/5}$. The following version of the Elias bound will be crucial in this section. 
\begin{theorem} \label{thm:elias}
    Let $1 \le d,n$ be integers and $\theta_{m}$ be the average Lee weight over $\mathbb{Z}_m$
For every $r \leq \theta_{m}n$ such that $$0<r^{2}-2\theta_{m}nr+\theta_{m}nd,$$ the following upper bound holds
$$A_{m}(n,d) \leq \frac{\theta_{m} nd}{r^{2}-2\theta_{m}nr + \theta_{m}nd} \cdot \frac{m^{n}}{\bbql{n,r}}.$$
\begin{proof}
     Similarly as in the Hamming metric, see \cite[Theorem 5.2.11]{van1998introduction}.
\end{proof}
\end{theorem}
\begin{proposition} \label{prop:Elias}
Let $\theta_{m}$ be the average Lee weight over $\Zm$. If $10 \sqrt{\theta_{m}n}  < t < n^{4/5},$ we have
$$A_{m}(n,2t+1) = o\left(\frac{H_{m}(n,t)}{n}\right).$$
\begin{proof}
    We apply the Elias bound with $r=t+\alpha$, where $\alpha = 7$. When $10 \sqrt{\theta_{m}n}  <t$  we satisfy the condition
    \begin{equation}\label{eq:EliasCond}
        r^2 \lesssim r^2 - \theta_{m}(2 \alpha -1)n = r^2-2 \theta_{m}nr + \theta_{m}n(2t+1). 
    \end{equation}
    In fact, when $10 \sqrt{\theta_{m}n} < t < n^{4/5}$ we have $\theta_{m}(2\alpha -1)n= 13 \cdot \theta_{m}n$ whereas $ 100 \cdot \theta_{m}n< r^2.$   
  Applying Theorem \ref{thm:elias}, together with the second part of Lemma \ref{lem:volumeestimate} gives
    \begin{align*} 
    A_{m}(n,2t+1) & \le \frac{\theta_{m}n(2t+1)}{r^2-2\theta_{m}nr+\theta_{m}n(2t+1)} \cdot \frac{m^n}{\bbql{n,r}}\\
    &\stackrel{\eqref{eq:EliasCond}}{\lesssim} \frac{nt}{(t+\alpha)^2} \cdot \frac{m^n}{\bbql{n,t+\alpha}}\\
    &  \lesssim \begin{cases}
\frac{n}{t}\left( \frac{(t+\alpha)(n-\alpha+1)}{2n(n-\alpha+1-t)}\right)^{\alpha}H_{m}(n,t) &  \, \textnormal{if $m$ is even}\\
\frac{n}{C_{m}}\left( \frac{(C_{m}+\alpha)(n-\alpha+1)}{n(n-\alpha+1-t)}\right)^{\alpha}H_{m}(n,t)  & \textnormal{if $m$ is odd and $t=\lceil \frac{m}{2}\rceil$,} \\
\frac{n}{t}\left( \frac{(t+\alpha)(n-\alpha+1)}{n(n-\alpha+1-t)}\right)^{\alpha}H_{m}(n,t)  & \textnormal{in all other cases,} 
\end{cases}
\end{align*}
where $0< C_{m}$ is a sufficiently large constant.
Similar calculations to those used for \cite[ Proposition 5.2]{dong2022number} lead in all three cases to $A_{m}(n,2t+1)=o(H_{m}(n,t)/n)$. 
\end{proof}
\end{proposition}
This now implies a result about the number of $t$-Lee-error-correcting codes when $10 \sqrt{\theta_{m}n}  < t < n^{4/5}$.
\begin{theorem} \label{thm:largerdist}
    Let $2 \le m,  10 \sqrt{\theta_{m}n}  < t < n^{4/5},$ then the number of $t$-Lee-error-correcting codes of length $n$ in $\Zm^{n}$ is asymptotically at most $2^{o(H_{m}(n,t))},$ for sufficiently large $n$.
    \begin{proof}
        By Proposition \ref{prop:Elias} the number of $t$-Lee-error-correcting codes is at most
        \begin{align*}
          \binom{m^{n}}{\leq A_{m}(n,2t+1)} \leq   2^{(\log_{2}(m)) \cdot n\cdot o(H_{m}(n,t)/n)} = 2^{o(H_{m}(n,t))}.
        \end{align*}
      
    \end{proof}
\end{theorem}
In Theorem \ref{thm:largerdist} we obtained different restrictions on $t$ than Dong et al. in  \cite[Theorem 1.1 (b)]{dong2022number} for the Hamming metric. In the Hamming metric the region of $t$ is $10\sqrt{n} < t \leq (1-q^{-1})n - C_{q}\sqrt{n \log(n)},$ for some constant $C_{q} >0$, such that the number of $t$-Lee-error-correcting codes in $\F_{q}^{n}$ is at most $2^{o(H_{q}(n,t))}$. 
This region for $t$ is slightly larger than our region. The less restricted region on $t$ for Hamming metric codes stems from the fact that the estimates on the ratios of the volumes given in \cite[Lemma 3.1]{dong2022number} also depend on the size of the ambient space.

\section{Counting Codes on Bipartite Graphs} \label{sec:bipartite}
In this section, we determine upper and lower bounds on the number of $t$-Lee-error-correcting codes by estimating isolated nodes on bipartite graphs. In particular, we define the density function of these codes of given cardinality within the set of codes of the same cardinality. Using the bipartite graph model, we can even consider linear codes over $\Zp$ (for $p$ prime). We determine the asymptotic behaviour of the density function and obtain density results for Lee metric codes that are optimal with respect to various bounds.
Let $1 \leq t \leq \lfloor \frac{nm-2}{4} \rfloor$, and consider the density function 
    \begin{equation} \label{eq:density}
    \delta_m^L(\Zm^{n},S,t) := \frac{|\{\mC \subseteq \Zm^{n}  :  |\mC|=S, \, d^{\textnormal{L}}(\mC) \ge 2t+1\}|}{|\{\mC \subseteq  \Zm^{n}  :  |\mC|=S\}|}
\end{equation}
of (possibly) nonlinear t-Lee-error-correcting codes in $\Zm^{n}$ and the density function of linear t-Lee-error-correcting codes in $\mathbb Z_p^n$ with minimum Lee distance at least $2t+1$, which is denoted by
    \begin{equation} \label{eq:densitylinear}
    \delta_p^L[\Zp^{n},S,t] := \frac{|\{\mC \leq \Zp^{n}  :  |\mC|=S, \, d^{\textnormal{L}}(\mC) \ge 2t+1\}|}{|\{\mC \leq  \Zp^{n}  :  |\mC|=S\}|}.
\end{equation}
Note that nonlinear Lee metric codes can have any cardinality $2 \leq |\mC| \leq m^n$, whereas the cardinality of linear codes is a power of $p$.

To give upper and lower bounds on the densities \eqref{eq:density} and \eqref{eq:densitylinear} we use counting arguments on the number of isolated nodes in \emph{bipartite graphs} that are regular with respect to certain maps defined on their left nodes. This approach was first proposed in \cite{gruica2022common} and applied to different metrics and different types of linearity in \cite{gruica2022densities}. 

We recall the concept of an \emph{association} on a finite non-empty set $\mV$ of magnitude $r \geq 0$ (see \cite{gruica2022common}), which is a function 
$\alpha: \mV \times \mV \to \{0,...,r\}$ satisfying the following:
\begin{itemize}
\item[(i)] $\alpha(V,V)=r$ for all $V\in \mV$;
\item[(ii)] $\alpha(V,V')=\alpha(V',V)$ for all $V,V' \in \mV$.
\end{itemize}

Given a finite bipartite graph $\mB=(\mV,\mW,\mE)$  and an association $\alpha$ on~$\mV$ of magnitude $r$, we call $\mB$ \emph{$\alpha$-regular} if for all  $(V,V') \in \mV \times \mV$ the number of nodes $W \in \mW$ with $(V,W) \in \mE$ and 
$(V',W) \in \mE$ depends only on $\alpha(V,V')$. We then denote this number by~$\mW_\ell(\alpha)$, where $\ell=\alpha(V,V') \in \{0, \dots , r \}$.
 \begin{proposition}\cite[Lemma 3.2 \& Lemma 3.5]{gruica2022common} \label{lem:upperbound}
 \\
Let $\mB=(\mV,\mW,\mE)$ be a bipartite $\alpha$-regular graph, where $\alpha$ is an association on $\mV$ of magnitude $r$. 
Let $\mF \subseteq \mW$ be the collection of non-isolated nodes of $\mW$. If $\mW_{r}(\alpha) > 0$, then
\begin{enumerate}
    \item [(i)]$|\mF| \le |\mV| \, \mW_{r}(\alpha).$
\item[(ii)]$|\mF| \ge  \frac{\mW_r(\alpha)^2 \, |\mV|^2}{\sum_{\ell=0}^r  \mW_\ell(\alpha) \, |\alpha^{-1}(\ell)|}.$
\end{enumerate}
\end{proposition}
\subsection{Nonlinear Lee Metric Codes}
We begin by giving bounds on the number of (possibly) nonlinear Lee metric codes in $\Zm^{n}$ whose minimum distance is bounded from below by $2t+1$.

\begin{theorem} \label{thm:nonlinearbounds}
    Let $S \geq 1$ and $1 \leq t <+\infty$ be integers. Define the quantities
    \begin{align*}
        \beta^{0} &:= \frac{1}{2}m^{n}(\bbql{n,2t}-1)-2\bbql{n,2t}+3,\\
        \beta^{1} &:=2\bbql{n,2t}-4,\\
        \Theta &:= 1+\beta^{1}\frac{S-2}{m^{n}-2}+\beta^{0}\frac{(S-2)(S-3)}{(m^{n}-2)(m^{n}-3)},
    \end{align*}
    and let $\mF := \{\mC \subseteq \Zm^{n} : |\mC| = S, d^{\textnormal{L}}(\mC) \leq 2t\}$. We have
    $$\frac{1}{2\Theta}m^{n}(\bbql{n,2t}-1)\binom{m^{n}-2}{S-2} \leq |\mF| \leq \frac{1}{2}m^{n}(\bbql{n,2t}-1)\binom{m^{n}-2}{S-2}.$$
    \begin{proof}
        The proof is analogous to  \cite[Theorem 2.1]{gruica2022densities}, using the two bounds from Proposition \ref{lem:upperbound}.
    \end{proof}
\end{theorem}
From Theorem \ref{thm:nonlinearbounds} we can derive the following asymptotic behaviour of the density function of Lee metric codes in $\Zm^{n}$.
\begin{theorem} \label{thm:nonlinear}
    Let $1 \leq t < +\infty$.
    \begin{enumerate}
        \item[i)] Let $n \geq 3$ be an integer. Consider the sequence $(H_{m}(n,2t+1))_{m \geq 2}$, where each $H_{m}(n,2t+1)$ is defined as in \eqref{eq:Hamming} and a sequence of integers $(S_m)_{m \geq 2}$,  with $S_m \geq 1$ for all $m \geq 2$. We have
    \begin{align*}
  \lim_{m \to +\infty}\delta_m^L(\Zm^{n},S_m,2t+1) =   \begin{cases}
 1 \quad &\textnormal{ if $S_{m} \in o(\sqrt{H_{m}(n,2t+1)})$ as $m \to +\infty$,}  \\
    0 \quad &\textnormal{ if $S_{m} \in \omega(\sqrt{H_{m}(n,2t+1)})$ as $m \to +\infty$.} 
    \end{cases}
    \end{align*}
            \item[ii)] Let $m \geq 2$. Consider the sequence $(H_{m}(n,2t+1))_{n \geq 1}$, where each $H_{m}(n,2t+1)$ is defined as in \eqref{eq:Hamming} and a sequence of integers $(S_n)_{n \geq 1}$,  with $S_n \geq 1$ for all $n \geq 1$. We have
    \begin{align*}
  \lim_{n \to +\infty}\delta_m^L(\Zm^{n},S_n,2t+1) =   \begin{cases}
 1 \quad &\textnormal{ if $S_{n} \in o(\sqrt{H_{m}(n,2t+1)})$ as $n \to +\infty$,}  \\
    0 \quad &\textnormal{ if $S_{n} \in \omega(\sqrt{H_{m}(n,2t+1)})$ as $n \to +\infty$.} 
    \end{cases}
\end{align*}
    \end{enumerate}

\end{theorem}
From Theorem \ref{thm:nonlinear} we obtain immediately the following result for the sphere covering bound (respectively Gilbert-Varshamov bound).
\begin{corollary}
    The probability that a uniformly chosen random Lee metric code in $\Zm^{n}$ satisfies the sphere covering bound tends to $0$ both as $m \to +\infty$ and $n \to +\infty$.
\end{corollary}
\subsection{Linear Lee Metric Codes}
In this section we consider linear codes $\mC \subseteq \Zp^{n}$ of a certain dimension $k \leq n$. Again we begin by giving bounds on the number of linear Lee metric codes in $\Zp^{n}$ whose minimum distance is bounded from below by $2t+1$.
We consider the same bipartite graph as done in the proof of \cite[Theorem 2.3]{gruica2022densities}, except that -- due to the fact that the Lee weight for an element $x \in \Zp^{n}$ may increase under scalar multiplication --
we consider the whole set of elements in $\Zp^{n}$ with Lee distance at most $2t$ (except zero) instead of considering $1$-dimensional vector spaces as done in \cite{gruica2022densities}. Hence, we let the left node set of the graph be $\mathbf{B}_{p}^{\textnormal{L}}(n,2t,0)\setminus \{0\}.$ Furthermore, we will consider linear Lee metric codes only over $\Zp$, i.e., fields of prime order. The reason for this is that the regularity conditions on the graph can no longer be fulfilled if the elements in the left node set of $\mB$ have different orders which is the case if we consider codes over $\Zps$ where $s>1.$

\begin{theorem} \label{cor:boundsLee}
Let $1 \le k \le n$ and $1 \le t < +\infty$ be integers and 
$\mF := \{ \mC\subseteq \Zp^{n} : \dim (\mC)=k, d^{\textnormal{L}}(\mC) \leq 2t\}.$ We have
$$\frac{(\bbqlp{n,2t}-1)\dstirling{n-1}{k-1}_{p}^{2}}{\dstirling{n-2}{k-2}_{p}(\bbqlp{n,2t}-p)+\dstirling{n-1}{k-1}_{p}(p-1)} \leq |\mF| \leq (\bbqlp{n,2t}-1)\dstirling{n-1}{k-1}_{p}.$$
\begin{proof}
    We consider the bipartite graph $\mB=(\mV,\mW,\mE)$, where $\mV = \mathbf{B}_{p}^{\textnormal{L}}(n,2t,0)\setminus \{0\}$, $\mW = \{\mC \subseteq \Zp^{n} : \dim (\mC) =k\}$  and 
  $\mE = \{(c,\mC) \in \mV \times \mW : c \in \mC \}$. 
 We have \begin{align*}
     |\mV| = \textbf{v}_p^{\textnormal{L}}(n,2t)-1, \quad  |\mW| = \dstirling{n}{k}_{p}.
 \end{align*}
It is easy to see that $\mB$ is left-regular of degree $$\dstirling{n-1}{k-1}_{p}.$$ Applying Proposition~\ref{lem:upperbound} we obtain an upper bound on $|\mF|$.

To prove the lower bound, we consider the association $$\alpha  :  \mV \times \mV \longrightarrow \{0,1\}, \quad (V,V') \mapsto 2-\dim\langle V,V' \rangle.$$ Here we obtain
\begin{align*}
       |\alpha^{-1}(0)|= |\mV| (|\mV|-(p-1)),\quad |\alpha^{-1}(1)|= |\mV|(p-1)
\end{align*}
and $\mB$ is $\alpha$-regular. Furthermore, we have
\begin{align*}
    \mW_0(\alpha)= \qbin{n-2}{k-2}{p}, \; \mW_1(\alpha)= \qbin{n-1}{k-1}{p},
\end{align*}
which -- combined with Proposition~\ref{lem:upperbound} -- implies the lower bound.
\end{proof}
\end{theorem}
As an immediate consequence we obtain the following bounds on the density function.
\begin{corollary} \label{cor:boundsdensitylin}
Let $1 \le k \le n$ and $1 \le t < +\infty$ be integers. We have
       \begin{align} \label{eq:upperBoundsublim}
       1-  \frac{(\bbqlp{n,2t}-1)\dstirling{n-1}{k-1}_{p} }{\dstirling{n}{k}_{p}} \leq \delta_p^L[\Zp^{n},p^{k},2t+1] &\leq 1-  \frac{(\bbqlp{n,2t}-1)\dstirling{n-1}{k-1}_{p}}{\bar\Theta\dstirling{n}{k}_{p}},
\end{align}
where $$\bar\Theta =  (p-1)+  \dstirling{n-1}{k -1}_{p}^{-1} (\bbqlp{n,2t}-p)  \dstirling{n-2}{k-2}_{p} .$$
\end{corollary}
From Corollary \ref{cor:boundsdensitylin} we derive the following asymptotic behaviour of the density of linear Lee metric codes in $\Zp^{n}$.
\begin{theorem} \label{thm:asymplinear}
    Let $1 \leq t < +\infty$ be an integer.
    \begin{enumerate}
        \item Let $1 \leq k < n$ be integers, then $\lim_{p \to +\infty} \delta_p^L[\Zp^{n},p^{k},2t+1] =1.$
        \item Let $p \in P$, fix a rate $R \in [0,1)$ such that $k(n) = Rn.$ Then $\lim_{n \to +\infty} \delta_p^L[\Zp^{n},p^{k(n)},2t+1] =1.$
    \end{enumerate}
    \begin{proof}
        Considering the lower bound given in \ref{eq:upperBoundsublim} and using the asymptotic estimates of the $q$-binomial coefficient we have
        $$\frac{(\bbqlp{n,2t}-1)\dstirling{n-1}{k-1}_{p} }{\dstirling{n}{k}_{p}} \sim \frac{\bbqlp{n,2t}p^{k}}{p^{n}},$$
        both as $p \to + \infty$ and as $n \to +\infty$. Now the two statements follow from Lemma \ref{lem:volumeLandau}.
    \end{proof}
\end{theorem}
Hence, when fixing $n,k$ and $t$ the proportion of linear t-Lee-error-correcting codes approaches $1$ as $p$ tends to infinity.
From the corollary above we can conclude, that most of the codes will have a large distance. 

However, codes that attain the bound in Theorem \ref{thm:Plotkin} are sparse as $n \to +\infty$ and as $p \to +\infty$. This had already been observed in \cite{byrne2023bounds}. However, since this is the sharpest known bound on the minimum Lee distance of linear codes, we will focus on this bound and confirm the results of \cite{byrne2023bounds} using our more general technique. 
The Plotkin bound is in the classical case a much looser bound than the Singleton bound. Its optimal linear codes are constant weight codes, while their nonlinear counterparts are called equidistant codes. They find several applications (e.g. \cite{cald,optical,digital,massey}, they are closely related to Steiner systems (e.g. \cite{steiner,braun}), and due to their connections to minimal codes \cite{min} they also find applications in secret sharing schemes (see e.g. \cite{secret}). The situation is somewhat different for the Lee metric: the Lee metric Plotkin bound, first introduced by Wyner and Graham \cite{wyner}, was later improved by Chiang and Wolf \cite{chiang1971channels} and finally by Byrne and Weger in \cite{byrne2023bounds}. The resulting Lee metric Plotkin bound has thus been studied more extensively than the Lee metric Singleton bound and is a much tighter bound. The constant Lee weight codes were first studied by Wood \cite{wood} and their characterization and density has been completed in \cite{byrne2023bounds}. Applications of such codes may be found in areas similar to those of classical codes.

\begin{theorem} 
    Let $1 \leq t \leq \lfloor \frac{nm-2}{4} \rfloor$ be an integer. 
    \begin{enumerate}
        \item[(i)] Let $1 \le n$ be an integer, then 
        $\lim_{p \to +\infty} \delta_p^L[\Zp^n,p^{n-  \frac{8t+4}{p+1}+1},t]=0.$
        \item[(ii)] Let $p \in P$, then $\lim_{n \to +\infty} \delta_p^L[\Zp^n,p^{n-\frac{8t+4}{p+1}+1},t]=0.$
    \end{enumerate}
     \begin{proof}
         Recalling the asymptotic estimates of the $q$-binomial coefficient from  \cite[Lemma 3.4]{gruica2022densities}
      it is easy to see that for the upper bound given in \eqref{cor:boundsLee} we have
      \begin{equation*}
         \frac{(\bbqlp{n,2t}-1)\dstirling{n-1}{k-1}_{p}}{\bar\Theta\dstirling{n}{k}_{p}} \sim \frac{(\bbqlp{n,2t}-1)}{p^{n-k}(p-1)+\bbqlp{n,2t}-p},
      \end{equation*}
      both as $p \to + \infty$ and $n \to +\infty$.
      Letting $k=n-\frac{8t+4}{p+1}+1$ gives
         \begin{equation*} 
          \lim_{p \to +\infty} \frac{(\bbqlp{n,2t}-1)}{p^{\frac{8t+4}{p+1}-1}(p-1)+\bbqlp{n,2t}-p} =1,
         \end{equation*}
         and 
               \begin{equation*} 
          \lim_{n \to +\infty} \frac{(\bbqlp{n,2t}-1)}{p^{\frac{8t+4}{p+1}-1}(p-1)+\bbqlp{n,2t}-p} =1.
         \end{equation*}
         Using Corollary \ref{cor:boundsdensitylin} the statement follows.
     \end{proof}
\end{theorem}
     For linear codes of rate $R$ we can reformulate the \emph{sphere covering} or \emph{Gilbert-Varshamov (GV) bound} to
$$ R \geq 1-\frac{1}{n}\log_{p}(\bbqlp{n,2t}), $$ which is equivalent to
\begin{equation} \label{eq:GV}
   \log_{p} (\bbqlp{n,2t}) \geq n(1-R).
\end{equation}
In the following we study the asymptotic behaviour of codes achieving the bound \eqref{eq:GV} with respect to $p$ or $n$.
\begin{theorem}
    Let $1 \leq t \leq \lfloor \frac{nm-2}{4} \rfloor$ be integers. Let $\mC \subseteq \Zp^{n}$ be a linear Lee metric code of rate $R=1-\frac{1}{n}\log_{p}(\bbqlp{n,2t}) $. 
    Then the probability that $\mC$ has minimum distance at least $2t+1$, i.e., the code satisfies the GV bound, approaches $1$, as $p\rightarrow +\infty$.
        \begin{proof}
       The statement follows from $k(n) =Rn <n$ and Theorem \ref{thm:asymplinear}.
    \end{proof}
\end{theorem}

If we consider $n$ going to infinity, then the probability that a linear Lee metric code, chosen uniformly at random attains the GV bound is upper bounded by $\frac{p-2}{p-1}$ and lower bounded by $0.$
However, if we add an $\varepsilon$-environment, we obtain:
\begin{theorem}
    Let $1 \leq t \leq \lfloor \frac{nm-2}{4} \rfloor$ be integers, $0 <\varepsilon$ and let $\mC \subseteq \Zp^{n}$ be a linear Lee metric code of rate $R=1-\frac{1}{n}\log_{p}(\bbqlp{n,2t})-\varepsilon $, 
    chosen uniformly at random. Then, the probability that $\mC$ has minimum distance at least $2t+1$ approaches $1$, for $n\rightarrow +\infty$.
    \begin{proof}
  Setting $k:= k(n) =Rn$ and considering the lower bound from \ref{eq:upperBoundsublim} we obtain
  $$\frac{(\bbqlp{n,2t}-1)\dstirling{n-1}{k-1}_{p}}{\dstirling{n}{k}_{p}} \sim \frac{1}{p^{\varepsilon n}}.$$
  For $n \to +\infty$ the statement follows.
    \end{proof}
\end{theorem}
The previous theorem matches the results on the asymptotic behaviour with respect to the Gilbert-Varshamov bound from \cite{byrne2022density}, where the analog was shown for general metrics arising from additive weights, including the Hamming, homogeneous and Lee metric.

\section{Concluding Remarks}\label{sec:comp}
From computational results in Python for small parameters, i.e., $m=4,5, t \le 3$ and $n \le 20$, we observe in all cases that the upper bound on the number of $t$-Lee-error-correcting codes obtained from the graph container is sharper than the upper bound obtained from the bipartite graph.
However, proving the stronger bound $2^{O(A_{m}(n, 2t+1))}$ whenever $2t+1 < (1-m^{-1}-c)n$, for any fixed $0 <c$ is much harder because one has to show that sets of size much larger than $A_{m}(n, 2t+1)$ must have many pairs of elements with distance at most $2t+1$. In the context of $H$-free graphs this stronger bound is a well-known conjecture of Kleitman and Winston in \cite{kleitman1982number}, where it served as a motivation for the development of the container algorithm. 
We think that this algorithm is still very interesting for studying the number of codes that attain a certain bound. Indeed, one could extend and transform these results to other bounds, i.e., the Singleton- or Plotkin bound and use then these bounds as a proxy for the size of the code when looking for saturation estimates that can be used in the container algorithm. However, this would go beyond the scope of this article and we were initially interested in the characterization of the intersection of two balls with equal Lee radii and distance $\ell$ between their respective centers. Since our upper bounds on the size of the intersection of two Lee metric balls only depend on the Lee volume for which explicit formulae are known, and hence can be computed fast, we suggest these results for considering covering properties of Lee metric codes.

\section{Acknowledgements}
Violetta Weger  is  supported by the European Union's Horizon 2020 research and innovation programme under the Marie Sk\l{}odowska-Curie grant agreement no. 899987.

\bibliographystyle{plain}
\bibliography{ourbib}

\appendix
\section{Appendix}
   \begin{lemma} \label{lem:appendix}
       Let $1 \le n,t$, $2 \le m, 0 \le i \le t$ and $0\leq r\leq m-1$ be integers such that $r=2t-m\lfloor \frac{2t}{m}\rfloor$. Then
       \begin{enumerate}
           \item $\left\lfloor \frac{2(t-1)}{m} \right\rfloor = \left\lfloor \frac{2t}{m} \right\rfloor  \textnormal{ if }r \geq 2.$
             \item If $m$ is an even number, $r$ is also even.
           \item If $m$ is an odd number and $r \le 1$ we have
            $$\binom{n-\lfloor \frac{2t}{m}\rfloor}{t-\lfloor \frac{2t}{m}\rfloor\left( \lfloor \frac{m}{2}\rfloor +1\right)-i} = 
  \binom{n-\lfloor \frac{2t}{m}\rfloor}{- \frac{2t-r(m+1)}{2m}-i}=\begin{cases}
      1 & \textnormal{if } \,  2t=m+1 \text{ and } i=0 \\
0 & \, \text{else}
  \end{cases}.$$
       \end{enumerate}
  \begin{proof}
      It is easy to see that 
      $$\left\lfloor \frac{2(t-1)}{m} \right\rfloor = \left\lfloor \frac{2t}{m} \right\rfloor  \textnormal{ if }r \geq 2.$$
       Furthermore, if $m$ is even (and $m':=m/2\in \mathbb N$) we have $r=2t-2m'\lfloor \frac{2t}{2m'}\rfloor = 2(t-m'\lfloor \frac{t}{m'}\rfloor)$, i.e., $r$ is also even.
       For $r=0$ we get $\frac{m}{2} \lfloor \frac{2t}{m}\rfloor = t$ and for $r=1$ we get $\frac{m}{2} \lfloor \frac{2t}{m}\rfloor = t-\frac{1}{2}$. For $m$ odd and $r <2$ we have 
  $$\left\lfloor\frac{m+2}{2}\right\rfloor \left\lfloor \frac{2t}{m}\right\rfloor = \frac{m+1}{2} \left\lfloor \frac{2t}{m}\right\rfloor = \frac{(2t-r)(m+1)}{2m} 
  = t + \frac{t}{m} - r\frac{m+1}{2m} \in \left\{t+\frac{t}{m}, t+\frac{2t-m-1}{2m}\right\},$$
  yielding the statement.
  \end{proof}
   \end{lemma}

 \begin{proof}[Proof of Lemma \ref{lem:volumeestimate}]
  We first determine the ratios of the volumes where both parameters $n$ and $t$ vary and begin with the case that $m$ is even. 
    We have
            \begingroup
\allowdisplaybreaks
      \begin{align*}
         &\frac{\bbql{n-1,t-1}}{\bbql{n,t}}\\
         &\leq  \frac{\sum_{i=0}^{\lfloor \frac{2(t-1)}{m}\rfloor}(-1)^{i}\binom{n-1}{i}\left( \sum_{j=0}^{t-1-\frac{mi}{2}}2^{j}\binom{n-1}{j}\binom{t-1-\frac{mi}{2}}{j}\right)}{\sum_{i=0}^{\lfloor \frac{2t}{m}\rfloor}(-1)^{i}\binom{n}{i}\left( \sum_{j=0}^{t-1-\frac{mi}{2}}2^{j+1}\binom{n}{j+1}\binom{t-\frac{mi}{2}}{j+1}\right)}\\
         & \leq \frac{\sum_{i=0}^{\lfloor \frac{2(t-1)}{m}\rfloor}(-1)^{i}\binom{n}{i}\left( \sum_{j=0}^{t-1-\frac{mi}{2}}2^{j}\binom{n-1}{j}\binom{t-1-\frac{mi}{2}}{j}\right)}{\sum_{i=0}^{\lfloor \frac{2t}{m}\rfloor}(-1)^{i}\binom{n}{i}\left( \sum_{j=0}^{t-1-\frac{mi}{2}}2^{j+1}\binom{n}{j+1}\binom{t-\frac{mi}{2}}{j+1}\right)}\\
         & \stackrel{(*)}{\leq} \frac{\sum_{i=0}^{\lfloor\frac{2(t-1)}{m}\rfloor}(-1)^{i}\binom{n}{i}\left( \sum_{j=0}^{t-1-\frac{mi}{2}}2^{j}\binom{n-1}{j}\binom{t-1-\frac{mi}{2}}{j}\right)}{\sum_{i=0}^{\lfloor \frac{2(t-1)}{m}\rfloor}(-1)^{i}\binom{n}{i}\cdot 2\left( \sum_{j=0}^{t-1-\frac{mi}{2}}2^{j}\binom{n-1}{j}\left(\frac{n}{j+1} \right)\binom{t-1-\frac{mi}{2}}{j} \left( \frac{t-\frac{mi}{2}}{j+1}\right)\right)}\\
        & \leq \frac{\sum_{i=0}^{\lfloor \frac{2(t-1)}{m}\rfloor}(-1)^{i}\binom{n}{i}\left( \sum_{j=0}^{t-1-\frac{mi}{2}}2^{j}\binom{n-1}{j}\binom{t-1-\frac{mi}{2}}{j}\right)}{\sum_{i=0}^{\lfloor\frac{2(t-1)}{m}\rfloor}(-1)^{i}\binom{n}{i}\left(\frac{2n}{t-\frac{mi}{2}} \right)\left( \sum_{j=0}^{t-1-\frac{mi}{2}}2^{j}\binom{n-1}{j}\binom{t-1-\frac{mi}{2}}{j} \right)}\\
         &\leq \frac{\sum_{i=0}^{\lfloor\frac{2(t-1)}{m}\rfloor}(-1)^{i}\binom{n}{i}\left( \sum_{j=0}^{t-1-\frac{mi}{2}}2^{j}\binom{n-1}{j}\binom{t-1-\frac{mi}{2}}{j}\right)}{\left(\frac{2n}{t}\right)\sum_{i=0}^{\lfloor \frac{2(t-1)}{m}\rfloor}(-1)^{i}\binom{n}{i}\left( \sum_{j=0}^{t-1-\frac{mi}{2}}2^{j}\binom{n-1}{j}\binom{t-1-\frac{mi}{2}}{j} \right)}\\
        &\leq \left( \frac{t}{2n}\right),
      \end{align*}
      \endgroup
where (*) can be derived from Lemma \ref{lem:binomprelim} and since either $\lfloor \frac{2t}{m}\rfloor= \lfloor \frac{2(t-1)}{m}\rfloor$ or $m\lfloor \frac{2t}{m}\rfloor = 2t$ (see  Lemma \ref{lem:appendix}). 
In the latter case the summand for  $i=\lfloor 2t/m \rfloor$ in the denominator is
equal to zero.

Iteratively this gives      
\begin{align*}
           \frac{\bbql{n,t}}{\bbql{n+i,t+i}} &\leq \left( \frac{t+i}{2(n+i)}\right) \cdot  \left( \frac{t+i-1}{2(n+i-1)}\right) \cdot \cdots \cdot \left( \frac{t+1}{2(n+1)}\right) \leq   \left( \frac{t+i}{2(n+i)}\right)^{i}.
      \end{align*}

For $m$ odd and all cases except $2t=m+1$
this gives together with Lemma \ref{lem:appendix}

      \begingroup
\allowdisplaybreaks
      \begin{align*}
          &\frac{\bbql{n-1,t-1}}{\bbql{n,t}}\\
          &\leq \frac{\sum_{i=0}^{t-1}\binom{n}{i}\sum_{j=0}^{\lfloor 2(t-1)/m\rfloor}(-2)^{j}\binom{n-1}{j}\binom{n-j-1}{t-j(\lfloor m/2 \rfloor +1)-i-1}}{\sum_{i=0}^{t}\binom{n+1}{i}\sum_{j=0}^{\lfloor 2(t-1)/m \rfloor}(-2)^{j}\binom{n}{j}\binom{n-j}{t-j(\lfloor m/2 \rfloor +1)-i}}\\
          &\stackrel{(*)}{=} \frac{\sum_{i=0}^{t-1}\binom{n+1}{i+1} \left( \frac{i+1}{n+1}\right)\sum_{j=0}^{\lfloor 2(t-1)/m\rfloor}(-2)^{j}\binom{n}{j}\left( \frac{n-j}{n}\right)\binom{n-j}{t-j(\lfloor m/2 \rfloor+1)-i}\left( \frac{t-j(\lfloor m/2 \rfloor+1)-i}{n-j}\right)}{\sum_{i=0}^{t}\binom{n+1}{i+1} \left( \frac{i+1}{n+1-i}\right)\sum_{j=0}^{\lfloor 2(t-1)/m \rfloor}(-2)^{j}\binom{n}{j}\binom{n-j}{t-j(\lfloor m/2 \rfloor+1)-i}}\\
          &\leq \frac{\sum_{i=0}^{t-1}\binom{n+1}{i+1}\left(\frac{i+1}{n+1}\right)\left(\frac{t-i}{n}\right)\sum_{j=0}^{\lfloor 2(t-1)/m \rfloor}(-2)^{j}\binom{n}{j}\binom{n-j}{t-j(\lfloor m/2 \rfloor+1)-i}}{\sum_{i=0}^{t-1}\binom{n+1}{i+1}\left( \frac{i+1}{n+1-i}\right)\sum_{j=0}^{\lfloor 2(t-1)/m \rfloor}(-2)^{j}\binom{n}{j}\binom{n-j}{t-j(\lfloor m/2 \rfloor+1)-i}}\\
          &\leq \frac{\sum_{i=0}^{t-1}\binom{n+1}{i+1}\left(\frac{i+1}{n+1}\right)\left(\frac{t}{n}\right)\sum_{j=0}^{\lfloor 2(t-1)/m \rfloor}(-2)^{j}\binom{n}{j}\binom{n-j}{t-j(\lfloor m/2 \rfloor+1)-i}}{\sum_{i=0}^{t-1}\binom{n+1}{i+1}\left( \frac{i+1}{n+1}\right)\sum_{j=0}^{\lfloor 2(t-1)/m \rfloor}(-2)^{j}\binom{n}{j}\binom{n-j}{t-j(\lfloor m/2 \rfloor+1)-i}}\\
          &=\frac{t}{n},
      \end{align*}
          \endgroup 
          where (*) follows from Lemma \ref{lem:binomprelim}.
 Iteratively this leads to     
\begin{align*}
           \frac{\bbql{n,t}}{\bbql{n+i,t+i}} &\leq \left( \frac{t+i}{n+i}\right) \cdot  \left( \frac{t+i-1}{n+i-1}\right) \cdot \cdots \cdot \left( \frac{t+1}{n+1}\right) \leq   \left( \frac{t+i}{n+i}\right)^{i}.
      \end{align*}
      For the case $2t=m+1$
      we have
            \begingroup
\allowdisplaybreaks
      \begin{align*}
          &\frac{\bbql{n-1,t-1}}{\bbql{n,t}}\\
          &\leq \frac{\sum_{i=0}^{t-1}\binom{n+1}{i}\binom{n-1}{t-i-1}}{\sum_{i=0}^{t-1}\binom{n+1}{i}\left(-2n + \left( \frac{n}{t-i}\right)\binom{n-1}{t-i-1} \right)+\binom{n+1}{t}(1-2n)}\\
          &=  \frac{\sum_{i=0}^{t-1}\binom{n+1}{i}\binom{n-1}{t-i-1}-2t\left(\sum_{i=0}^{t-1}\binom{n+1}{i} + \binom{n+1}{t}\right)+\binom{n+1}{t}\left(\frac{t}{n}\right)+2t\left(\sum_{i=0}^{t-1}\binom{n+1}{i} + \binom{n+1}{t}\right)-\binom{n+1}{t}\left(\frac{t}{n}\right)}{\left( \frac{n}{t}\right)\left(\sum_{i=0}^{t-1}\binom{n+1}{i}\binom{n-1}{t-i-1}-2t\left(\sum_{i=0}^{t-1}\binom{n+1}{i} + \binom{n+1}{t}\right)+\binom{n+1}{t} \left(\frac{t}{n} \right) \right)}\\
          &= \frac{t}{n}+\frac{\frac{2t^2}{n}\left(\sum_{i=0}^{t-1}\binom{n+1}{i} + \binom{n+1}{t}\right)-\binom{n+1}{t}\left(\frac{t^2}{n^2}\right)}{\sum_{i=0}^{t-1}\binom{n+1}{i}\binom{n-1}{t-i-1}-2t\left(\sum_{i=0}^{t-1}\binom{n+1}{i} + \binom{n+1}{t}\right)+\binom{n+1}{t}\left(\frac{t}{n} \right)}.
      \end{align*}
          \endgroup
          Moreover
          \begin{align*}
              &\frac{\frac{2t^2}{n}\left(\sum_{i=0}^{t-1}\binom{n+1}{i} + \binom{n+1}{t}\right)-\binom{n+1}{t}\left(\frac{t^2}{n^2}\right)}{\sum_{i=0}^{t-1}\binom{n+1}{i}\binom{n-1}{t-i-1}-2t\left(\sum_{i=0}^{t-1}\binom{n+1}{i} + \binom{n+1}{t}\right)+\binom{n+1}{t}\left(\frac{t}{n} \right)}\\
              &\stackrel{(*)}{\leq} \frac{\frac{2t^2}{n}\left(\sum_{i=0}^{t-1}\left(\frac{e(n+1)}{i}\right)^{i} +\left(\frac{e(n+1)}{t}\right)^{t}\right)-\left(\frac{n+1}{t}\right)^{t}\left(\frac{t^2}{n^2}\right)}{\sum_{i=0}^{t-1}\left(\frac{n+1}{i}\right)^{i}\left(\frac{n-1}{t-i-1}\right)^{t-i-1}-2t\left(\sum_{i=0}^{t-1}\left(\frac{e(n+1)}{i}\right)^{i} + \left(\frac{e(n+1)}{t}\right)^{t}\right)+ \left(\frac{n+1}{t}\right)^{t}\left(\frac{t}{n} \right)}
              \\& \in O(n^{-1}),
          \end{align*}
          where (*) uses Lemma \ref{lem:binomprelim}. Hence
          $$\frac{\bbql{n-1,t-1}}{\bbql{n,t}} \leq \frac{C_{m}}{n},$$
          for a sufficiently large constant $C_{m}.$     
      Consequently,
      \begin{align*}
           \frac{\bbql{n,t}}{\bbql{n+i,t+i}} &= \frac{\bbql{n,t}}{\bbql{n+1,t+1}} \cdot \frac{\bbql{n+1,t+1}}{\bbql{n+2,t+2}} \cdots \frac{\bbql{n+i-1,t+i-1}}{\bbql{n+i,t+i}}\\
           &\leq \left( \frac{C_{m}}{n+i}\right)^{i}.
      \end{align*}
      This proves the first statement.

        We will now determine the ratios where only parameter $t$ varies.
        
           If $m$ is an even number we have
    \begin{align*}
        \frac{\bbql{n, t}}{\bbql{n-i, t}} &= \frac{\sum_{k=0}^{\lfloor 2t/m\rfloor}(-1)^{k}\binom{n}{k}\mathlarger{\mathlarger{\sum}}_{j=0}^{t-\lfloor m/2 \rfloor k}2^{j}\binom{n}{j}\binom{t-\lfloor m/2 \rfloor k}{j}}{\sum_{k=0}^{\lfloor 2t/m\rfloor}(-1)^{k}\binom{n-i}{k}\mathlarger{\mathlarger{\sum}}_{j=0}^{t-\lfloor m/2 \rfloor k}2^{j}\binom{n-i}{j}\binom{t-\lfloor m/2 \rfloor k}{j}}\\
        &\stackrel{(*)}{\le} \frac{\sum_{k=0}^{\lfloor 2t/m\rfloor}(-1)^{k}\binom{n}{k}}{\sum_{k=0}^{\lfloor 2t/m\rfloor}(-1)^{k}\left( \frac{n-i+1-k}{n-i+1}\right)\cdot \cdots \cdot \left( \frac{n-k}{n}\right) \binom{n}{k}\cdot \left(\frac{n-i+1-t+\lfloor m/2 \rfloor k}{n-i+1} \right)^{i}}\\
        &\leq \left( \frac{n-i +1}{n-i +1-t}\right)^{i},
    \end{align*}
    where (*) is an iterative application of Lemma \ref{lem:binomprelim}.
    Combining this with 
    $$\bbql{n-i,t}\leq \left(\frac{t+i}{2n}\right)^{i}\bbql{n,t+i}$$
    from Lemma \ref{lem:volumeestimate} gives the final statement.

    For the case that $m$ is an odd number we maximize 
    $$\left(\frac{n-i+1}{n-i+1-j} \right)^{i}\left(\frac{n-i-j+1}{n-i-j+1-t+j(\lfloor m/2 \rfloor +1)+k} \right)^{i}\left(\frac{n-i+2}{n-i+2-k} \right)^{i}$$
    under the constraints $0 \leq j \leq \lfloor 2t/m \rfloor$ and $0 \leq k \leq t$ and obtain
     \begin{align*}
        \frac{\bbql{n, t}}{\bbql{n-i, t}} = \frac{\sum_{k=0}^{t}\binom{n+1}{k}\mathlarger{\mathlarger{\sum}}_{j=0}^{\lfloor 2t/m\rfloor}(-2)^{j}\binom{n}{j}\binom{n-j}{t-j(\lfloor m/2 \rfloor +1)-k}}{\sum_{k=0}^{t}\binom{n-i+1}{k}\mathlarger{\mathlarger{\sum}}_{j=0}^{\lfloor 2t/m\rfloor}(-2)^{j}\binom{n-i}{j}\binom{n-i-j}{t-j(\lfloor m/2 \rfloor+1)-k}}
        \leq \left( \frac{n-i +1}{n-i+1-t}\right)^{i}.
    \end{align*}
    Using the first part of Lemma \ref{lem:volumeestimate} we derive the statement.
 \end{proof}
 \begin{proof}[Proof of Lemma \ref{lem:supersat01}] \label{lem:supersat01_proof}
 For $x \in \Zm^n$, let $K_{x}:=\{y \in \mC : d^{\textnormal{L}}(x,y) \le t\}$. Then
 $$\sum_{r=1}^{2t}W(t,r)|E_{r}|= \sum_{x \in \Zm^{n}}\binom{|K_{x}|}{2},$$
 since both terms count pairs $(x,\{y,z\}),$ where $x \in \Zm^{n}$ and $y,z \in \mC$ are distinct such that $d^{\textnormal{L}}(x,y),d^{\textnormal{L}}(x,z) \le t$.
Lemmata \ref{lem:mon} and \ref{lem:estimateinter} give
\begin{align} \label{eq:0}
    W(t,r) \le W(t,1) \le m\bbql{n-1,t-1} \le \begin{cases}
\frac{mt \bbql{n,t}}{2n} & \,\textnormal{if $m$ is even}, \\
\frac{mC_{m} \bbql{n,t}}{n} & \, \textnormal{if $m$ is odd and $t= \lceil \frac{m}{2}\rceil$},\\
\frac{mt \bbql{n,t}}{n} & \, \textnormal{in all other cases},\\
\end{cases}
\end{align}
for $1 \le r \le 2t$. The average value of $K_{x}$ over all $x \in \Zm^{n}$ is $\frac{|\mC|\bbql{n,t}}{m^n}.$ Thus, by the lower bound given in Lemma \ref{lem:binomprelim} we have
\begin{equation} \label{eq:1}
    \frac{|\mC|^{2}\bbql{n,t}^{2}}{10m^n} \le m^n \binom{|\mC|\bbql{n,t}/m^{n}}{2} \le \sum_{x \in \Zm^{n}}\binom{|K_{x}|}{2} = \sum_{r=1}^{2t}W(t,r)|E_{r}|.
\end{equation}
This proves the first inequality of the Lemma. Observing that
$ \sum_{r=1}^{2t}W(t,r)|E_{r}| \le W(t,1)|E|$, rearranging and combing the two inequalities \eqref{eq:0} and \eqref{eq:1}, the second inequality can be derived.

 \end{proof}
 \begin{proof}[Proof of Lemma \ref{lem:C1}]
  Let $v \in \mC_1$, the overlap of $\ballL{n,t,v}$ with other balls $\ballL{n,t,u}, u \in \mC_{1}$, has size
        \begin{align} \label{eq:sumoverlap}
            \left| \ballL{n,t,v} \cap \bigcup_{u \in \mC_{1}\setminus \{v\}}\ballL{n,t,u} \right| & \leq \sum_{r=1}^{2t}\textnormal{deg}_{r}(v) W(t,r) \leq \left( \sum_{r=1}^{20}\textnormal{deg}_{r}(v)W(t,r)\right)+ n^{5}W(t,21),
        \end{align}
        by Lemma \ref{lem:mon}. From Lemma \ref{lem:estimateinter} we have $W(t,r) \lesssim n^{-\lceil r/2 \rceil /2} \bbql{n,t}$ in all cases.
        This gives $n^{5}W(t,21) \lesssim \epsilon \bbql{n,t}$, for sufficiently large $n$. Moreover, using the definition of $\mC_{1}$, i.e., $\textnormal{deg}_{r}(v) \le \varepsilon n^{\lceil r/2 \rceil /2}$ for each $v \in \mC$, $1 \le r\le 20$ the sum \eqref{eq:sumoverlap} is $\lesssim \varepsilon \bbql{n,t}$. Thus, for every $v \in \mC_{1}$, the ball $\ballL{n,t,v}$ contains $(1-O(\varepsilon))\bbql{v,t}$ unique points not shared by other such balls, and thus the union of these balls has size $\geq (1- O(\varepsilon))\bbql{n,t}|\mC_{1}|.$ Since this union is contained in $\Zm^{n}$, we derive $|\mC_{1}| \leq (1+O(\varepsilon))H_{m}(n,t)$ by definition of the Hamming bound.
        \end{proof}
 \begin{proof}[Proof of Lemma \ref{lem:C2}]
            Pick a subset of maximum size $X \subseteq \mC_{2}$, where every pair of distinct elements of $X$ has a Lee distance greater than $t$. For each $x \in X$, one can find an $\lceil \log(n)/\varepsilon \rceil$ element subset $A_{x} \subseteq \mC \cap \ballL{n,20,x}$. For distinct $x,y \in X$ one has $d(x,y) > t \geq 60$, and so $A_{x} \cap A_{y} = \emptyset$.

       The intersection of one of the balls $\ballL{n,t,u}, u \in \bigcup_{x \in X}A_{x}$  with the union of all other such balls has size at most 
        \begin{align*}
        \frac{\log(n)}{\varepsilon}W(t,1) + n^{5}W(t,21) &\le \begin{cases*}
\left( \frac{\log(n)mt}{\varepsilon 2n} + \frac{n^5 t^{11}}{2n^{11}}\right)\bbql{n,t} &  \textnormal{if $m$ is even,}\\
 \left( \frac{\log(n)m C_{m}}{\varepsilon n} + \frac{n^5 C_{m}^{11}}{n^{11}}\right)\bbql{n,t} &  \textnormal{if $m$ is odd and $t=\lceil \frac{m}{2}\rceil$,} \\
 \left( \frac{\log(n)mt}{\varepsilon n} + \frac{n^5 t^{11}}{n^{11}}\right)\bbql{n,t} &  \textnormal{else,}
\end{cases*}\\
       & = o(\bbql{n,t}),
        \end{align*}
        as $n \to +\infty$ and $C_{m}$ is chosen sufficiently large.
      In fact, for each of the $\lceil \log(n)/\varepsilon\rceil -1$ points $\Tilde{u}$ that are in the same $A_{x}$ as $u$, the overlap is $\leq W(t,1)$ and each ball $\ballL{n,t,\Tilde{u}}$ with $\Tilde{u}$
in some $A_y$ other than $A_x$ the intersection with $\ballL{n,t,u}$ contains at most $\leq W(t,21)$ elements. Since $W(t,r)=0$ if $r >2t$ there are at most $\Delta(G[S]) \leq n^{5}$ other balls intersecting a given ball.
      So the union of the balls $\ballL{n,t,u}, u \in \bigcup_{x \in X}A_{x}$ has size $|X|\lceil \log(n)/\varepsilon \rceil \cdot (1-o(1))\bbql{n,t}.$ Since this quantity is at most $m^{n}$ this gives
      $$|X| \leq \frac{m^{n}}{\lceil \log(n)/\varepsilon \rceil \cdot (1-o(1))\bbql{n,t}} \lesssim \frac{\varepsilon H_{m}(n,t)}{\log(n)}.$$
      Note that $\bigcup_{x\in X}\ballL{n,t,x}= \Zm^{n}$ and each pair of elements of $\ballL{n,t,x}$ has Lee distance at most $2t$. Since independent sets in $G[S_{2}]$ are formed by elements which have Lee distance greater than $2t$, one can obtain an independent set by choosing at most one element from $\ballL{n,t,x} \cap S_{2}$ for each $x \in X$. Since $|\ballL{n,t,x} \cap S_{2}| \leq n^{5}$, this gives
      \begin{align*}
          i(G[S_{2}]) \leq n^{5|X|} \leq 2^{O(\varepsilon H_{m}(n,t))}. 
      \end{align*}
      \end{proof}
       \begin{proof}[Proof of Lemma \ref{lem:supersat2b}]
For $x \in \Zm^n$, let $K_{x}:=\{y \in \mC : d^{\textnormal{L}}(x,y) \le t\}$ and for $r \in \mathbb{N}$ set $S_{r} := \{ x \in \Zm^n : |K_{x}|=r\}$. Then $m^n + \gamma \bbql{n,t} \le |\mC|\bbql{n,t} = \sum_{r \ge 1} r |S_{r}|.$ So the number of pairs in $\mC$ of distance at most $2t$ is at least
$$ \frac{\gamma \bbql{n,t}}{m\bbql{n-1,t-1}} \le \frac{\gamma \bbql{n,t}}{W(t,1)} \le \frac{1}{W(t,1)}\sum_{r \ge 1}|S_{r}|(r-1) \le \frac{1}{W(t,1)}\sum_{r \ge 1}|S_{r}|\binom{r}{2}.$$
    Now the statement can be derived from Lemma \ref{lem:volumeestimate}.
      \end{proof}
        \begin{proof}[Proof of Lemma \ref{lem:func1}]
         Using structural induction we observe that $I \subseteq P \cup f(P)$. When the algorithm terminates we know that $\Delta (G_{i}) < \Delta$ and hence $|E(G_{i})| < (|V(G_{i})|\Delta/2)$. Thus Lemma \ref{lem:supersat2b} gives us $|f(P)|= |V(G_{i})| \leq  (1+\varepsilon)H_{m}(n,t)$. To prove that $|P| \leq \varepsilon H_{m}(n,t)\log(n)/n$  we distinguish two stages of the algorithm, according to the size of $V(G_{i})$.
    \begin{itemize}
        \item Let $P_{1}$ denote the set of nodes $u \in P$ added to $P$ when $|V(G_{i-1})| \geq 2H_{m}(n,t) $.
           \item Let $P_{2}$ denote the set of nodes $u \in P$ added to $P$ when $|V(G_{i-1})| < 2H_{m}(n,t).$
    \end{itemize}
     As we add nodes to $P_1$, by Lemma \ref{lem:supersat01} we remove at least a $\beta \gtrsim n/H_{m}(n,t)$ fraction of nodes from $G_{i-1}$ to get $G_{i}$.
 Thus $(1-\beta)^{k}m^{n} \leq 2H_{m}(n,t)$. By Lemma \ref{lem:volumeLandau} we have $n/H_{m}(n,t) \rightarrow 0$ as $n \to +\infty$, such that $\beta \leq \log \left( \frac{1}{1-\beta}\right)$ for sufficiently large $n$ and therefore 
    $$|P_{1}| \leq \frac{\log\left( \frac{m^{n}}{2H_{m}(n,t)}\right)}{\log \left( \frac{1}{1-\beta}\right)}\le \frac{\log(\bbql{n,t)}}{\beta}\stackrel{(*)}{\lesssim} \frac{\log(n)H_{m}(n,t)}{n},$$
where (*) follows from Lemmata \ref{lem:binomprelim} and \ref{lem:leebounds}.
   After step $k$ we remove at least $\varepsilon n$ vertices at each step, $0<\varepsilon$ is chosen sufficiently small. So we have 
    $$|P_{2}| \leq \frac{2H_{m}(n,t)}{\varepsilon n}.$$

This gives $$|P|= |P_1|+|P_2| \lesssim \frac{ H_{m}(n,t)\log(n)}{n}.$$
Moreover
$$|P \cup f(P)| \lesssim \frac{ H_{m}(n,t)\log(n)}{n} + (1+ \varepsilon)H_{m}(n,t),$$
and thus
 $$i(G[P \cup f(P)]) \leq 2^{|P|} \cdot i(f(P)) \leq 2^{(1+2\varepsilon)H_{m}(n,t)}.$$
      \end{proof}
\end{document}